\documentclass[lettersize,journal]{IEEEtran}
\usepackage{amsmath,amsfonts}
\usepackage{algorithmic}
\usepackage{algorithm}
\usepackage{array}
\usepackage[caption=false,font=normalsize,labelfont=sf,textfont=sf]{subfig}
\usepackage{textcomp}
\usepackage{stfloats}
\usepackage{url}
\usepackage{verbatim}
\usepackage{graphicx}
\usepackage{cite}
\usepackage{amssymb}
\usepackage[utf8]{inputenc} 
\usepackage[T1]{fontenc}
\usepackage{ifthen}
\usepackage{amsmath,amsthm}
\usepackage{tikz}
\usepackage{tikz-3dplot}
\usetikzlibrary{backgrounds,shapes.geometric}
\usepackage{pgfplots}
\pgfplotsset{compat=1.18}
\usetikzlibrary{calc}
\usepackage[justification=centering]{caption}
\usepackage{multirow}
\usepackage{bm}

\newtheorem{thm}{Theorem}[section]

\newtheorem{lemma}[thm]{Lemma}

\newtheorem{example}[thm]{Example}
\newtheorem*{example*}{Example}

\newtheorem*{remark*}{Remark}

\usepackage{bm}
 
\newcommand{\bldg}{{\bm{g}}}
\newcommand{\bldu}{{\bm{u}}}
\newcommand{\bldv}{{\bm{v}}}
\newcommand{\bldc}{{\bm{c}}}
\newcommand{\bldx}{{\bm{x}}}
\newcommand{\bldy}{{\bm{y}}}
\newcommand{\blde}{{\bm{e}}}
\newcommand{\code}{{\mathcal{C}}}
\newcommand{\ico}{{\text{\normalfont ico}}}
\newcommand{\dod}{{\text{\normalfont dod}}}

\begin{document}
\title{A New Class of Geometric Analog Error Correction Codes for Crossbar Based In-Memory Computing} 

\author{%

\IEEEauthorblockN{%
\begin{tabular}{c}
Ziyuan Zhu \\
University of California, San Diego \\
\texttt{ziz050@ucsd.edu}
\end{tabular}
\hspace{1.5em}
\begin{tabular}{c}
Changcheng Yuan \\
Texas A\&M University \\
\texttt{ericycc@tamu.edu}
\end{tabular}
\hspace{1.5em}
\\[0.8em]
\begin{tabular}{c}
Ron M. Roth \\
Technion \\
\texttt{ronny@cs.technion.ac.il}
\end{tabular}
\hspace{1.5em}
\begin{tabular}{c}
Paul H. Siegel \\
University of California, San Diego \\
\texttt{psiegel@ucsd.edu}
\end{tabular}
\hspace{1.5em}
\begin{tabular}{c}
Anxiao Jiang \\
Texas A\&M University \\
\texttt{ajiang@cse.tamu.edu}
\end{tabular}
}

}

\maketitle

\begin{abstract}
Analog error correction codes have been proposed for analog in-memory computing on resistive crossbars, which can accelerate vector–matrix multiplication for machine learning. Unlike traditional communication or storage channels, this setting involves a mixed noise model with small perturbations and outlier errors. A number of analog codes have been proposed for handling a single outlier, and several constructions have also been developed to address multiple outliers. However, the set of available code families remains limited, covering only a narrow range of code lengths and dimensions.
In this paper, we study a recently proposed family of geometric codes capable of handling multiple outliers, and develop a geometric analysis that characterizes their $m$-height profiles.
\end{abstract}

\section{Introduction}
Analog in-memory computing is a cutting-edge technology that integrates data storage and computation directly within memory cells, enabling significant acceleration of deep neural network (DNN) computations~\cite{Sebastian2020, Wang2022}. The primary motivation for analog in-memory computing is to overcome the "von Neumann bottleneck" by avoiding the need for massive data transfers between processors and memory \cite{Wong2015, Gallo2023, Zhang2023}. This approach promises substantial improvements in speed and energy efficiency by exploiting the vector--matrix multiplications within crossbar architectures \cite{Roth2022}, a fundamental operation in DNNs. 

Analog ECCs have been proposed to address the challenge of computing against errors~\cite{Roth2020, Roth2023}. These codes are specifically designed to handle errors in codewords transmitted through channels that introduce two primary types of additive noise: limited-magnitude errors (LMEs) and unlimited-magnitude errors (UMEs). LMEs are small but widespread, arising from effects like cell-programming noise and random read/write disturbances in non-volatile memory (NVM) arrays. In contrast, UMEs, such as stuck cells or short cells, occur less frequently but can be significantly more disruptive. While DNNs can often tolerate minor, distributed noise, they are particularly vulnerable to large, isolated errors, making robust error correction essential for reliable in-memory computation~\cite{Sebastian2020}. 

Some existing analog ECC designs primarily target detection of a single UME~\cite{Roth2020, Roth2023}, while others address multiple UMEs~\cite{Wei2024}. 
In addition, some promising short codes that handle multiple UMEs have been obtained via genetic algorithms~\cite{Anxiao_LP}. More recently, \cite{Roth2026} proposed a new class of geometric codes including polygonal codes, dual polygonal codes, polyhedral codes, and dual polyhedral codes together with a linear-programming analysis of their error-handling capability. 
In this paper, we develop an alternative geometric analysis for the dual polygonal codes and dual polyhedral codes, and characterize their $m$-height profiles, recovering the same results as in \cite{Roth2026}. Extending this geometric approach to polygonal codes and polyhedral codes remains an open problem. 

\section{Preliminaries}
\label{sec:preliminary}

\subsection{Height profile of analog ECC}

Let \(\mathcal{C}\) be a linear \([n,k]\) code over \(\mathbb{R}\) (namely, a $k$-dimensional subspace of \(\mathbb{R}^n\)) that is generated by (the rows of) a real $k\times n$ matrix $G=[\bldg_0,\dots,\bldg_{n-1}]$, where $\bldg_i$ is a $k \times 1$ vector. Each codeword has the form $\bldc=\bldu G$, where $\bldu\in\mathbb{R}^k$ is an arbitrary input vector and $\bldc\in\mathbb{R}^n$. For a codeword $\bldc=(c_0,c_1,\ldots,c_{n-1})\in\mathcal{C}$, consider the magnitudes $|c_0|,|c_1|,\ldots,|c_{n-1}|$. 
Let
\[
|c_{(0)}|\ge |c_{(1)}|\ge\cdots\ge |c_{(n-1)}|
\]
denote the elements of $\{|c_j|\}_{j=0}^{n-1}$ sorted in nonincreasing order.

The $m$-height of a nonzero codeword $\bldc$ is defined as the ratio of its largest-magnitude entry to its $(m+1)$-th largest-magnitude entry, i.e.,
\[
h_m(\bldc)=\frac{|c_{(0)}|}{|c_{(m)}|}.
\]
If $c_{(m)}=0$, then $h_m(\bldc)$ is defined to be $+\infty$.  
The $m$-height of the code $\mathcal{C}$ is then
\[
h_m(\mathcal{C})=\max_{\bldc\in\mathcal{C}\setminus\{\boldsymbol{0}\}} h_m(\bldc).
\]
Theorem~2 of~\cite{Anxiao_LP} gives a general method to compute $h_m(\mathcal{C})$ by solving a family of linear programs, and a simplified version of this approach is provided in \cite{Roth2026}.  
 Define \([n\rangle\) as \(\{0, 1, \dots, n-1\}\).
The sequence of $m$-heights $h_m(\mathcal{C}), m \in [n\rangle$, is called the height profile of the code $\mathcal{C}$.

\subsection{Error model and error-handling capability of analog ECC}

Let $\delta$ and $\Delta$ be positive real thresholds satisfying $\Delta > \delta > 0$.
An error vector $\boldsymbol{\varepsilon}=(\varepsilon_0,\varepsilon_1,\ldots,\varepsilon_{n-1})\in\mathbb{R}^n$
is referred to as an LME vector if $\varepsilon_i\in[-\delta,\delta]$ for all $i\in[n\rangle$.
Let $\blde=(e_0,e_1,\ldots,e_{n-1})\in\mathbb{R}^n$ denote a UME vector representing outliers.  
A noisy received word $\bldy=(y_0,\ldots,y_{n-1})\in\mathbb{R}^n$ is given by $\bldy=\bldc+\boldsymbol{\varepsilon}+\blde$.

The vector $\blde$ is assumed to be sparse, since it models outlier faults caused by rare events such as stuck cells or shorted cells in the array.
Unlike Gaussian noise, such faults occur infrequently, but when they do occur, they typically have large magnitudes.  
An outlier is defined by \( |e_i|>0 \). Among such outliers, \textbf{significant} outliers are those satisfying \( |e_i|>\Delta \).

Following~\cite{Roth2020, Roth2023}, we define
the support of an outlier vector $\blde$  as  
\[
\operatorname{Supp}_{0}(\blde):=\{i\in [n\rangle: |e_i|>0\}.
\]
Similarly, 
define the $\Delta$-support of $\blde$ as
\[
\operatorname{Supp}_{\Delta}(\blde):=\{i\in [n\rangle: |e_i|>\Delta\}.
\]
The error-handling capability of analog ECC differs conceptually from that of classical error-correcting codes over finite fields because of the nature of the analog error model. For nonnegative integers $\tau, \sigma$, not both zero,  we say that a length-$n$ analog code $\mathcal{C}$ is $\tau$-error locating and $(\tau+\sigma)$-error detecting if there is a decoder
\[
\mathcal{D}: \mathbb{R}^n  \rightarrow  \mathcal{C} \; \cup \; \{\text{``e''}\}
\]
with the following property.

\begin{enumerate}
\renewcommand{\labelenumi}{(\roman{enumi})}
\item $\operatorname{Supp}_{\Delta}({\blde})
\subseteq
\mathcal{D}({\bldy})
\subseteq
\operatorname{Supp}_{0}(\blde),
\quad
\text{if} \; | \operatorname{Supp}_{0}(\blde)|\leq  \tau .$

\item $\mathcal{D}(\bldy) {=}  \text{``e''}
\;  \text{or}  \;
\operatorname{Supp}_{\Delta}(\blde)
\subseteq
\mathcal{D}(\bldy)\neq  \text{``e''},
\quad
\text{otherwise}.$

\end{enumerate}

In other words,  if $\blde$ has no more than $\tau$ UMEs, the decoder output is a subset of the support of $\blde$ that contains all significant UMEs. Otherwise, the decoder output is either the error flag ``e'' or a set of positions that includes all significant outlying errors of $\blde$. Note that in the latter case, the set may include ``false alarms'' that are positions outside the support of $\blde$.

The proof of Theorem~1 in~\cite{Roth2020, Roth2023} introduced a particular decoder for locating $\tau$ significant UMEs and detecting $\tau+\sigma$ UMEs. Specifically, if   $\blde$ is a UME vector with $|\operatorname{Supp}_{0}(\blde)| \le \tau+\sigma$, then the decoder $\mathcal{D}$ operates as follows:
\begin{enumerate}
\renewcommand{\labelenumi}{(\roman{enumi})}
\item The output $\mathcal{D}(\bldy)$ is ``e'' if there is no compatible UME vector $\blde'$ such that $|\operatorname{Supp}_{0}(\blde')| \le \tau$.

\item Otherwise, the output $\mathcal{D}(\bldy)$ is a set of size at most $\tau$, defined as the intersection of the supports of all compatible UME vectors $\blde'$ with $|\operatorname{Supp}_{0}(\blde')| \le \tau$. By Lemma~2 in~\cite{Roth2023}, this set is guaranteed to contain $\operatorname{Supp}_{\Delta}(\blde)$.
\end{enumerate}

The following example further illustrates the behavior of the decoder for the $[n,1]$ repetition code over $\mathbb{R}$ discussed in Example~1 of~\cite{Roth2023}.

\begin{example}
Consider the repetition code
\[
\mathcal{C}=\{(a,a,\dots,a)\in\mathbb{R}^{12}: a\in\mathbb{R}\},
\]
with parameters
\[
n=12,\quad k=1,\quad \delta=1,\quad
\Delta=4,\quad \tau=5,\quad \sigma=1.
\]
Assume that the transmitted input is $u=0$, so the transmitted codeword is
\[
\bldc=(0,0,\dots,0).
\]
The received vector is
\[
\bldy=\bldc+\blde+\boldsymbol{\varepsilon},
\]
where $|\varepsilon_i|\le 1$ for all $i\in [12\rangle$.

We consider the following three cases.

\noindent\textbf{(i) A case with $|\operatorname{Supp}_{0}(\blde)|\le \tau$.}
Let
\[
\blde=
(5,\,4.5,\,2.6,\,2,\,
0,\,0,\,0,\,0,\,
0,\,0,\,0,\,0),
\]
and let $\boldsymbol{\varepsilon}=\boldsymbol{0}$. Then
\[
\bldy=
(5,\,4.5,\,2.6,\,2,\,
0,\,0,\,0,\,0,\,
0,\,0,\,0,\,0),
\]
and
\[
\operatorname{Supp}_{0}(\blde)=\{0,1,2,3\},
\qquad
\operatorname{Supp}_{\Delta}(\blde)=\{0,1\}.
\]
The possible compatible codewords $\bldc'=(a,\dots,a)$ with a compatible UME vector $\blde'$ satisfying
\[
|\operatorname{Supp}_{0}(\blde')|\le 5
\]
are those with $a\in[-1,1]$.

If $a\in[-1,1)$, then
\[
\operatorname{Supp}_{0}(\blde')=\{0,1,2,3\}.
\]
If $a=1$, then
\[
\operatorname{Supp}_{0}(\blde')=\{0,1,2\}.
\]
Hence, the possible supports are
\[
\{0,1,2,3\}
\qquad \text{and} \qquad
\{0,1,2\},
\]
and therefore
\[
\mathcal{D}(\bldy)=\{0,1,2\}.
\]

\noindent\textbf{(ii) A case with $\tau<|\operatorname{Supp}_{0}(\blde)|\le \tau+\sigma$.}
Let
\[
\blde=
(5,\,5,\,4.6,\,4.3,\,
2.7,\,2.4,\,0,\,0,\,
0,\,0,\,0,\,0),
\]
and let $\boldsymbol{\varepsilon}=\boldsymbol{0}$. Then
\[
\bldy=
(5,\,5,\,4.6,\,4.3,\,
2.7,\,2.4,\,0,\,0,\,
0,\,0,\,0,\,0),
\]
and
\[
\operatorname{Supp}_{0}(\blde)=\{0,1,2,3,4,5\},
\]
\[
\operatorname{Supp}_{\Delta}(\blde)=\{0,1,2,3\}.
\]
There is no compatible codeword $\bldc'=(a,\dots,a)$ such\\
that the corresponding compatible UME vector $\blde'$ satisfies
\[
|\operatorname{Supp}_{0}(\blde')|\le 5.
\]
Therefore,
\[
\mathcal{D}(\bldy)=\text{``e''}.
\]

\noindent\textbf{(iii) A case with $|\operatorname{Supp}_{0}(\blde)|>\tau+\sigma$.}
Let
\[
\blde=
(5,\,5,\,5,\,5,\,
5,\,5,\,5,\,3,\,
3,\,3,\,3,\,3),
\]
and let $\boldsymbol{\varepsilon}=\boldsymbol{0}$. Then
\[
\bldy=
(5,\,5,\,5,\,5,\,
5,\,5,\,5,\,3,\,
3,\,3,\,3,\,3),
\]
and
\[
|\operatorname{Supp}_{0}(\blde)|=12>\tau+\sigma=6.
\]
Taking
\[
\bldc'=(4,4,\dots,4),
\]
we have
\[
|y_i-4|\le 1,\qquad \forall i\in [12\rangle,
\]
so there exists a compatible decomposition with
\[
\blde'=\boldsymbol{0}.
\]
Hence,
\[
\mathcal{D}(\bldy)=\varnothing.
\]
\end{example}

\section{Dual Polygonal Codes}
\label{sec:2D_code}

In this section, we construct a class of $k=2$ analog codes whose generator matrices consist of evenly spaced unit vectors over a half circle, which is the same as the code in Example 4 of~\cite{Roth2026}.

Let $k=2$ and let $G=[\bldg_0,\dots,\bldg_{n-1}]\in\mathbb{R}^{2\times n}$ consist of unit columns
\[
\bldg_j=\begin{bmatrix}\cos\theta_j\\[2pt]\sin\theta_j\end{bmatrix},
\qquad \theta_j\in[0,\pi).
\]
Choose $\theta_j=\frac{\pi}{n}\,j$, $j=0,1,\dots,n-1$. Then the set $\{\bldg_j\}$ forms an evenly spaced set of vectors over a half circle. An example with $n=3$ is illustrated in Fig.~\ref{fig:half_circle_example}. The dashed lines mark the antipodal directions of the generator vectors.

\begin{figure}[htbp]
\centering
\begin{tikzpicture}[scale=1.2]
  \draw (0,0) circle (1.5);

  \fill[red!25] (0,0) -- (0:1.5) arc (0:30:1.5) -- cycle;

  \foreach \j/\ang in {0/0,1/60,2/120}{
    \draw[->] (0,0) -- (\ang:1.5) node[pos=1.15] {$\bldg_{\j}$};
  }

  \foreach \ang in {180,240,300}{
    \draw[dashed] (0,0) -- (\ang:1.5);
  }
\end{tikzpicture}
\caption{\centering An example of dual polygonal codes for $n=3$.}
\label{fig:half_circle_example}
\end{figure}
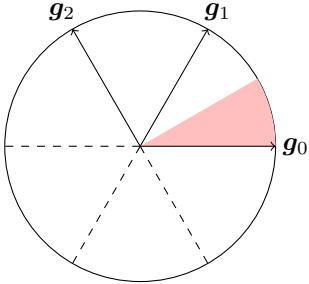

For a unit information vector $\bldu(\alpha)=(\cos\alpha,\sin\alpha)^\top$, the resulting codeword is $\bldc(\alpha)=\bldu(\alpha)^\top G\in\mathbb{R}^n$, whose entries are $c_j(\alpha)=\cos(\theta_j-\alpha)$. Let $c_{(0)}(\alpha)\ge c_{(1)}(\alpha)\ge\cdots\ge c_{(n-1)}(\alpha)$ denote the elements of $\{|c_j(\alpha)|\}_{j=0}^{n-1}$ sorted in nonincreasing order. Then the $m$-height of the code is
\[
h_m(\mathcal{C})=\sup_{\alpha\in[0,2\pi)}\frac{c_{(0)}(\alpha)}{c_{(m)}(\alpha)}.
\]

Information directions $\alpha\in\bigl[0,\tfrac{\pi}{2n}\bigr]$ generate a subset of codewords. Due to symmetry, all other codewords can be obtained from this subset by suitable permutations and sign changes of the coordinates, which do not affect the $m$-height. As a result, the interval $\alpha\in\bigl[0,\tfrac{\pi}{2n}\bigr]$ constitutes a complete and sufficient domain for analyzing the $m$-height of the code, as illustrated in Fig.~\ref{fig:half_circle_example}. 
Hence, the $m$-height simplifies to
\[
h_m(\mathcal{C})
= \sup_{\alpha\in[0,2\pi)}\frac{c_{(0)}(\alpha)}{c_{(m)}(\alpha)}
= \sup_{\alpha\in[0,\frac{\pi}{2n}]}\frac{c_{(0)}(\alpha)}{c_{(m)}(\alpha)}.
\]

\begin{lemma}
\label{lemma:Order_2D}
Let $n\ge 2$ and fix $\alpha\in\bigl[0,\tfrac{\pi}{2n}\bigr]$.
Then the order statistics $c_{(k)}(\alpha)$ of $\lvert c_j(\alpha)\rvert$ are attained at the indices
\[
c_{(k)}(\alpha)=|c_{j_k}(\alpha)|,\qquad k=0,1,\ldots,n-1,
\]
where
\[
j_k=
\begin{cases}
0, & k=0,\\[4pt]
\dfrac{k+1}{2}, & \text{$k$ odd},\\[6pt]
n-\dfrac{k}{2}, & \text{$k$ even and } k\ge2.
\end{cases}
\]
\end{lemma}

\begin{proof}
Recall that
\[
c_j(\alpha)=\cos\!\Bigl(\tfrac{\pi j}{n}-\alpha\Bigr),\qquad j=0,1,\ldots,n-1,
\]
and by definition $c_{(k)}(\alpha)$ denotes the $(k+1)$-th largest value among
$\bigl\{\lvert c_j(\alpha)\rvert\bigr\}_{j=0}^{n-1}$.

Since
$| \cos x|$ is $\pi$-periodic and even, it is convenient to reduce to the principal strip $(-\tfrac{\pi}{2},\tfrac{\pi}{2}]$.
Define
\[
P(x) \;:=\; x-\pi\Big\lfloor\frac{x+\tfrac{\pi}{2}}{\pi}\Big\rfloor
\in\bigl(-\tfrac{\pi}{2},\tfrac{\pi}{2}\bigr],
\]
so that
\[
\bigl|\cos P(x)\bigr| = \bigl|\cos x\bigr|
\quad\text{for all }x\in\mathbb{R}.
\]
Let
\[
\begin{aligned}
\mathcal{M}
&:= \bigl\{-\alpha + j\tfrac{\pi}{n} : j=0,1,\ldots,n-1\bigr\},\\
\mathcal{L}
&:= P(\mathcal{M})
 = \bigl\{\,P(-\alpha + j\tfrac{\pi}{n}) : j=0,1,\ldots,n-1\,\bigr\}.
\end{aligned}
\]
Then
\[
\begin{aligned}
& \bigl \{\lvert c_j(\alpha)\rvert\bigr\}_{j=0}^{n-1}
 = \bigl\{\lvert\cos x\rvert : x\in\mathcal{M}\bigr\}\\
&= \bigl\{\lvert\cos x\rvert : x\in\mathcal{L}\bigr\} 
 = \bigl\{\cos x : x\in\mathcal{L}\bigr\},
\end{aligned}
\]
because $\mathcal{L}\subset(-\tfrac{\pi}{2},\tfrac{\pi}{2}]$ and $\cos x>0$
on this interval.

On $[0,\tfrac{\pi}{2}]$ the function $x\mapsto\cos x$ is strictly decreasing,
and $\cos x$ is even. Hence for $|\alpha_1|,|\alpha_2|\le\tfrac{\pi}{2}$ we have
\[
\cos\alpha_1 \ge \cos\alpha_2
\iff |\alpha_1|\le|\alpha_2|.
\]
Therefore, within $\mathcal{L}$, ordering the values
$\{\cos x:x\in\mathcal{L}\}$ in nonincreasing order is equivalent to ordering
the absolute values $\{|x|:x\in\mathcal{L}\}$ in nondecreasing order.

We now describe $\mathcal{L}$ explicitly. First, for
$j=0,1,\ldots,\lfloor\tfrac{n}{2}\rfloor$ we have
\[
-\alpha+j\tfrac{\pi}{n}\in\bigl(-\tfrac{\pi}{2},\tfrac{\pi}{2}\bigr],
\]
\[
P\!\Bigl(-\alpha+j\tfrac{\pi}{n}\Bigr)
= -\alpha + j\tfrac{\pi}{n}.
\]
These contribute the points
\[
-\alpha,\ \tfrac{\pi}{n}-\alpha,\ \tfrac{2\pi}{n}-\alpha,\ 
\tfrac{3\pi}{n}-\alpha,\ \ldots
\]
up to the last $j$ for which the expression stays in
$(-\tfrac{\pi}{2},\tfrac{\pi}{2}]$.

Next, consider $j=\tfrac{n}{2}+r$ (if $n$ is even) or
$j=\tfrac{n-1}{2}+1+r$ (if $n$ is odd); in either case we can write
\[
j = n-r,\qquad r=1,2,\ldots,\Bigl\lfloor\tfrac{n-1}{2}\Bigr\rfloor.
\]
Then
\[
-\alpha + j\tfrac{\pi}{n}
= -\alpha + (n-r)\tfrac{\pi}{n}
= \pi - \Bigl(\tfrac{\pi r}{n}+\alpha\Bigr),
\]
which lies in $(\tfrac{\pi}{2},\pi)$ for the relevant $r$. Applying $P$ gives
\[
P\!\Bigl(-\alpha + (n-r)\tfrac{\pi}{n}\Bigr)
= -\Bigl(\tfrac{\pi r}{n}+\alpha\Bigr),
\]
so these contribute the points
\[
-\Bigl(\tfrac{\pi}{n}+\alpha\Bigr),\ 
-\Bigl(\tfrac{2\pi}{n}+\alpha\Bigr),\ 
-\Bigl(\tfrac{3\pi}{n}+\alpha\Bigr),\ \ldots.
\]

Collecting both parts, the elements of $\mathcal{L}$ in
$(-\tfrac{\pi}{2},\tfrac{\pi}{2}]$ can be written as
\[
-\alpha,\ \tfrac{\pi}{n}-\alpha,\ \tfrac{2\pi}{n}-\alpha,\ \ldots
\quad\text{and}\quad
-\Bigl(\tfrac{\pi}{n}+\alpha\Bigr),\ 
-\Bigl(\tfrac{2\pi}{n}+\alpha\Bigr),\ \ldots
\]
hence their absolute values form the interleaving
sequence
\[
\alpha,\ 
\tfrac{\pi}{n}-\alpha,\ 
\tfrac{\pi}{n}+\alpha,\ 
\tfrac{2\pi}{n}-\alpha,\ 
\tfrac{2\pi}{n}+\alpha,\ 
\tfrac{3\pi}{n}-\alpha,\ 
\tfrac{3\pi}{n}+\alpha,\ \ldots,
\]
because for every integer $r\ge1$ and every
$\alpha\in[0,\tfrac{\pi}{2n}]$,
\[
r\tfrac{\pi}{n}-\alpha
\;\le\;
r\tfrac{\pi}{n}+\alpha
\;\le\;
(r+1)\tfrac{\pi}{n}-\alpha.
\]

We now map these distances back to indices $j$. By construction,
\begin{itemize}
  \item The smallest distance $\alpha$ corresponds to $j=0$.
  \item The next distance $\tfrac{\pi}{n}-\alpha$ corresponds to $j=1$.
  \item The distance $\tfrac{\pi}{n}+\alpha$ corresponds to $j=n-1$.
  \item The distance $\tfrac{2\pi}{n}-\alpha$ corresponds to $j=2$.
  \item The distance $\tfrac{2\pi}{n}+\alpha$ corresponds to $j=n-2$.
\end{itemize}
and so on. In general, the interleaving pattern of distances yields the
index sequence
\[
0,\ 1,\ n-1,\ 2,\ n-2,\ 3,\ n-3,\ \ldots,
\]
i.e.,
\[
j_k=
\begin{cases}
0, & k=0,\\[4pt]
\dfrac{k+1}{2}, & \text{$k$ odd},\\[6pt]
n-\dfrac{k}{2}, & \text{$k$ even and } k\ge2.
\end{cases}
\]

Since $c_{(k)}(\alpha)$ is, by definition, the $(k+1)$-th largest element of
$\{\lvert c_j(\alpha)\rvert\}_{j=0}^{n-1}$, and we have just shown that
the corresponding indices appear in the order $j_0,j_1,j_2,\ldots$ given
above, which proves the lemma.
\end{proof}

Thus, by Lemma~\ref{lemma:Order_2D}, the $m$-height computation can be simplified as follows.
\[
h_m(\mathcal{C})
=\sup_{\alpha\in[0,\frac{\pi}{2n}]}\frac{c_{(0)}(\alpha)}{c_{(m)}(\alpha)}
=\sup_{\alpha\in[0,\frac{\pi}{2n}]}\frac{\cos\alpha}{\cos\theta_{j_m}}.
\]

For the dual polygonal code, since $k=2$, the Singleton bound implies that  there exists a codeword with $c_{(n-1)}=0$, so 
$h_{n-1}( \mathcal{C})=\infty$. 
Moreover, 
the direction $\bldu(\alpha)$ can be orthogonal to at most one generator
vector $\bldg_j$, that is, there is at most one $j$ such that
$\cos(\theta_j-\alpha)=0$. Hence at most the smallest order statistic
$c_{(n-1)}(\alpha)$ can be zero, implying that $\mathcal{C}$ is MDS.   In particular, for every $m\le n-2$, the denominator $c_{(m)}(\alpha)$ is strictly positive for all
$\alpha\in[0,\tfrac{\pi}{2n}]$, and the corresponding $m$-height is finite.
We therefore restrict attention to $m\le n-2$. 

\begin{thm}
\label{thm:hm_max_halfcircle_simple}
For the polygonal codes, and $0 < m \le n-2$, the $m$-height achieves its maximum at
\[
\arg\max_{\alpha\in[0,\frac{\pi}{2n}]} h_m(\alpha)
  =
  \begin{cases}
    \alpha=\tfrac{\pi}{2n}, & \text{$m$ even},\\[2pt]
    \alpha=0,              & \text{$m$ odd}.
  \end{cases}
\]
Moreover, the $m$-height is
\[
h_m(\mathcal{C})
=
\begin{cases}
\displaystyle \dfrac{\cos \tfrac{\pi}{2n}}{\cos\bigl((m+1)\tfrac{\pi}{2n}\bigr)}, & \text{$m$ even},\\[10pt]
\displaystyle \dfrac{1}{\cos\bigl((m+1)\tfrac{\pi}{2n}\bigr)},      & \text{$m$ odd}.
\end{cases}
\]
\end{thm}

\begin{proof}
From Lemma~\ref{lemma:Order_2D}, the ordered term for $0 < m \le n-2$ satisfies
\[
c_{(m)}(\varphi)=\cos\!\big(A_{(m)}+s_{(m)}\varphi\big),
\]
where
\[
A_{(m)}=\frac{\pi}{n}\Bigl\lceil\frac{m}{2}\Bigr\rceil,
\qquad
s_{(m)}=(-1)^m,
\]
and therefore
\[
h_m(\varphi)=\frac{\cos\varphi}{\cos\!\big(A_{(m)}+s_{(m)}\varphi\big)}.
\]
Taking a derivative and simplifying,
\[
h_m'(\varphi)
=\frac{s_{(m)}\sin A_{(m)}}{\cos^2\!\big(A_{(m)}+s_{(m)}\varphi\big)}.
\]
Since $A_{(m)}\in(0,\tfrac{\pi}{2}]$, the numerator has the same sign as $s_{(m)}$,
and the denominator is strictly positive. Thus,
\[
\operatorname{sgn}\,h_m'(\varphi)
=\operatorname{sgn}(s_{(m)})
=
\begin{cases}
+1,& m~\text{even},\\
-1,& m~\text{odd}.
\end{cases}
\]
Hence $h_m(\varphi)$ is strictly increasing when $m$ is even and strictly
decreasing when $m$ is odd, proving
\[
\arg\max h_m
=
\begin{cases}
\varphi=\tfrac{\pi}{2n},& m~\text{even},\\
\varphi=0,& m~\text{odd}.
\end{cases}
\]

Finally, evaluating at the maximizing point:
\[
h_m(\mathcal{C})
=
\begin{cases}
\displaystyle \dfrac{\cos \tfrac{\pi}{2n}}{\cos\bigl((m+1)\tfrac{\pi}{2n}\bigr)}, & \text{$m$ even},\\[10pt]
\displaystyle \dfrac{1}{\cos\bigl((m+1)\tfrac{\pi}{2n}\bigr)},      & \text{$m$ odd}.
\end{cases}
\]
This completes the proof.
\end{proof}

\section{Dual Polyhedral Codes}
\label{sec:Dual_solid_code}

In this section, we introduce two codes derived from three-dimensional geometric structures: the icosahedron and the dodecahedron. These codes are related to the dual icosahedral code $\code_\ico^\perp$ and dual duodecahedral code $\code_\dod^\perp$ introduced in Examples~5 and~6 of \cite{Roth2026}, respectively.  The codes we analyze differ from those constructions by at most permutations and sign changes of columns and/or rows of the generator matrix, which will not affect the $m$-height profile.

\subsection{Dual icosahedral codes}

The icosahedron has $12$ vertices, $20$ faces, and $30$ edges. By placing its vertices on the sphere, we obtain $6$ symmetric axes $\{\boldsymbol{g}_{\boldsymbol{0}}, \ldots, \boldsymbol{g}_{\boldsymbol{5}}\}$, where each axis corresponds to a pair of antipodal vertices. These six axes can be represented by the following matrix:
\[
G =
\begin{bmatrix}
0 & 0 & 1 & 1 & \varphi & \varphi \\
1 & 1 & \varphi & -\varphi & 0 & 0 \\
\varphi & -\varphi & 0 & 0 & 1 & -1
\end{bmatrix}, \qquad
\varphi = \frac{1+\sqrt{5}}{2}.
\]

When searching for the optimal information vector $\bldu$ that attains the $m$-height, 
only its direction matters (i.e., its magnitude does not affect the $m$-height.)
Therefore, it is sufficient to restrict the analysis to the faces of the solid, which cover all possible directions.
Among the $20$ faces, all are symmetric, and furthermore, each face contains subregions that are 
themselves symmetric. Thus, analyzing a single representative subregion is equivalent to analyzing 
the entire space.

Motivated by this symmetry, we restrict our attention to the subregion whose direction is closest to 
\(\bldg_0\), as illustrated in Fig.~\ref{fig:icos}. We define the smaller triangle as
\[
T := \mathrm{conv}\{\bldv_0,\bldv_1,\bldv_2\}
\]
with
\[
\bldv_0 := \bldg_0,\qquad
\bldv_1 := \frac{\bldg_0+\bldg_2}{2},\qquad
\bldv_2 := \frac{\bldg_0+\bldg_2+\bldg_4}{3}.
\]


\begin{figure}

\tdplotsetmaincoords{70}{65}
\centering
\begin{tikzpicture}[tdplot_main_coords, scale=2.5, rotate=30,
  line cap=round, line join=round]

\begin{scope}[xshift=-2.6cm]

\coordinate (p1)  at (0.,0.,-0.9510565162951536);
\coordinate (p2)  at (0.,0., 0.9510565162951536);
\coordinate (p3)  at (-0.85065080835204,0.,-0.42532540417601994);
\coordinate (p4)  at ( 0.85065080835204,0., 0.42532540417601994);
\coordinate (p5)  at ( 0.6881909602355868,-0.5,-0.42532540417601994);
\coordinate (p6)  at ( 0.6881909602355868, 0.5,-0.42532540417601994);
\coordinate (p7)  at (-0.6881909602355868,-0.5, 0.42532540417601994);
\coordinate (p8)  at (-0.6881909602355868, 0.5, 0.42532540417601994);
\coordinate (p9)  at (-0.2628655560595668,-0.8090169943749475,-0.42532540417601994);
\coordinate (p10) at (-0.2628655560595668, 0.8090169943749475,-0.42532540417601994);
\coordinate (p11) at ( 0.2628655560595668,-0.8090169943749475, 0.42532540417601994);
\coordinate (p12) at ( 0.2628655560595668, 0.8090169943749475, 0.42532540417601994);

\def\FrontFaces{
{1,5,9},{1,5,6},{4,6,12},{2,4,12},{2,7,11},
{4,5,11},{5,9,11},{7,9,11},{2,11,4},{4,5,6}
}

\def\BackFaces{
{2,12,8},{2,8,7},{5,9,1},{6,5,1},
{10,6,1},{3,10,1},{9,3,1},{12,10,8},{8,3,7}
}

\foreach \poly in {
{3,10,1},
{12,10,8},
}{
  \draw[gray!40, ultra thin]
    plot[samples at=\poly, variable=\x] (p\x) -- cycle;
}


\draw[gray!40, ultra thin] (p12) -- (p8) -- (p2);

\draw[gray!40, ultra thin] (p2) -- (p8) -- (p7);

\draw[gray!40, ultra thin] (p6) -- (p10) -- (p1);

\draw[gray!40, ultra thin] (p9) -- (p3) -- (p1);

\coordinate (v2) at ($(p4)!0.5!(p11)$); 
\coordinate (v3) at ($0.333333*(p4)+0.333333*(p5)+0.333333*(p11)$); 

\filldraw[fill=gray!35, draw=black, line width=0.6pt]
  (p4) -- (v2) -- (v3) -- cycle;

\draw[black, line width=0.9pt] (p4) -- (v2) -- (v3) -- cycle;

\fill (v2) circle[radius=0.8pt];
\fill (v3) circle[radius=0.8pt];

\draw[gray!40, ultra thin]
  plot[samples at={8,3,7}, variable=\x] (p\x) -- cycle;

\foreach \a/\b/\c in {1/5/9, 1/5/6, 4/6/12, 2/4/12, 2/7/11,
                     4/5/11, 5/9/11, 7/9/11, 2/11/4, 4/5/6}{
  \draw[very thick, fill=white, fill opacity=0.15]
    (p\a) -- (p\b) -- (p\c) -- cycle;
}

\fill (p4)  circle[radius=1.2pt];
\fill (p5)  circle[radius=1.2pt];
\fill (p11) circle[radius=1.2pt];

\node[font=\footnotesize, xshift=8pt, yshift=2pt]  at (p4)  {$\bldg_0$};
\node[font=\footnotesize, xshift=10pt, yshift=1.5pt]  at (p5)  {$\bldg_4$};
\node[font=\footnotesize, xshift=-2.7pt, yshift=10pt] at (p11) {$\bldg_2$};

\node[font=\footnotesize, xshift=-7pt, yshift=3pt] at (v2) {$\bldv_1$};
\node[font=\footnotesize, xshift=0pt, yshift=-7pt] at (v3) {$\bldv_2$};

\end{scope}
\end{tikzpicture}
\caption{Icosahedron with a shaded triangular region indicating the fundamental search space.}
\label{fig:icos}
\end{figure}
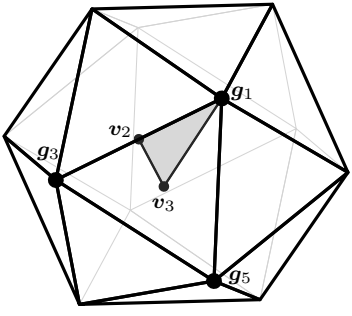

Any \(\bldx\in T\) can be written in barycentric form
\[
\bldx(u,v) := u \bldv_0 + v \bldv_1 + (1-u-v)\bldv_2.
\]
The parameter domain is the standard triangle
\[
D := \{(u,v)\in\mathbb{R}^2 : u\ge 0,\ v\ge 0,\ u+v\le 1\}.
\]

\begin{lemma}
\label{lemma:order_icos}
For every \(\bldx\in T\), the six absolute values \(|\bldx\cdot \bldg_i|\), \(i=0,\dots,5\), satisfy the global order
\[
|\bldx\cdot \bldg_0|
\;\ge\;
|\bldx\cdot \bldg_2|
\;\ge\;
|\bldx\cdot \bldg_4|
\;\ge\;
|\bldx\cdot \bldg_3|
\;\ge\;
|\bldx\cdot \bldg_1|
\;\ge\;
|\bldx\cdot \bldg_5|.
\]
\end{lemma}

\begin{proof}
A direct computation with the parametrization \(\bldx(u,v)\) yields
\begin{align*}
\bldx\cdot \bldg_0 &= \frac{4u+v}{3} + \frac{\sqrt{5}}{2} + \frac{7}{6},\\
\bldx\cdot \bldg_2 &= \frac{-2u+v}{3} + \frac{\sqrt{5}}{2} + \frac{7}{6},\\
\bldx\cdot \bldg_4 &= \frac{-2u-2v}{3} + \frac{\sqrt{5}}{2} + \frac{7}{6},
\end{align*}
and
\begin{align*}
\bldx\cdot \bldg_1 &= -\frac{\varphi}{3}\,(2u - v + 1),\\
\bldx\cdot \bldg_3 &= -\frac{\varphi}{3}\,(2u + 2v + 1),\\
\bldx\cdot \bldg_5 &= -\frac{\varphi}{3}\,(4u + v - 1).
\end{align*}
From these expressions we first obtain, for all \((u,v)\in T\),
\[
\bldx\cdot \bldg_0>0,\quad \bldx\cdot \bldg_2>0,\quad \bldx\cdot \bldg_4>0,
\quad
\bldx\cdot \bldg_3<0,
\quad
\bldx\cdot \bldg_1\le 0,
\]
and \(\bldx\cdot \bldg_5\) changes sign across the line \(4u+v=1\).

For \(\bldg_0,\bldg_2,\bldg_4\), we have
\[
\bldx\cdot \bldg_0 - \bldx\cdot \bldg_2 = 2u\ge 0,\qquad
\bldx\cdot \bldg_2 - \bldx\cdot \bldg_4 = v\ge 0,
\]
with equality only on the edges \(u=0\) and \(v=0\), respectively. Since these three quantities are always positive, this implies
\[
|\bldx\cdot \bldg_0|\ge|\bldx\cdot \bldg_2|\ge|\bldx\cdot \bldg_4|
\]
for all \(\bldx\in T\).

To compare \(|\bldx\cdot \bldg_4|\) with \(|\bldx\cdot \bldg_3|\), note that \(\bldx\cdot \bldg_4>0\) and \(\bldx\cdot \bldg_3<0\), so
\[
|\bldx\cdot \bldg_4|\ge|\bldx\cdot \bldg_3|
\quad\Longleftrightarrow\quad
\bldx\cdot \bldg_4 + \bldx\cdot \bldg_3 \ge 0.
\]
Using the formulas above, one finds
\[
\bldx\cdot \bldg_4 + \bldx\cdot \bldg_3 
= \frac{\sqrt{5}+3}{3}\,(1-u-v),
\]
which is nonnegative on \(T\) and becomes zero exactly when \(u+v=1\). Similarly, since \(\bldx\cdot \bldg_1\le 0\) and \(\bldx\cdot \bldg_3<0\), we compare
\[
|\bldx\cdot \bldg_3|\ge|\bldx\cdot \bldg_1|
\quad\Longleftrightarrow\quad
\bldx\cdot \bldg_3 - \bldx\cdot \bldg_1 \le 0,
\]
and a direct computation shows
\[
\bldx\cdot \bldg_3 - \bldx\cdot \bldg_1 = -\varphi\,v,
\]
which is nonpositive on \(T\), with equality only when \(v=0\). For \(|\bldx\cdot \bldg_1|\) and \(|\bldx\cdot \bldg_5|\), 
\[
|\bldx\cdot \bldg_1|\ge |\bldx\cdot \bldg_5|
\ \Longleftrightarrow\
2u-v+1 \ \ge\ |\,4u+v-1\,|.
\]
Split into two cases:

\emph{(i) If $4u+v\ge1$,} then $|4u+v-1|=4u+v-1$ and
\[
2u-v+1 \ge 4u+v-1
\ \Longleftrightarrow\
2 \ge 2u+2v
\ \Longleftrightarrow\
1\ge u+v,
\]
which holds on $T$, with equality exactly on the edge $u+v=1$.

\emph{(ii) If $4u+v\le1$,} then $|4u+v-1|=1-4u-v$ and
\[
2u-v+1 \ge 1-4u-v
\ \Longleftrightarrow\
6u \ge 0,
\]
always true, with equality exactly on the edge $u=0$.

Combining (i)–(ii), we obtain \(|\bldx\cdot \bldg_1|\ge |\bldx\cdot \bldg_5|\) for all \((u,v)\in T\), and equality holds only on \(u=0\) or \(u+v=1\).

Combining these comparisons yields the global chain
\[
|\bldx\cdot \bldg_0|
\;\ge\;
|\bldx\cdot \bldg_2|
\;\ge\;
|\bldx\cdot \bldg_4|
\;\ge\;
|\bldx\cdot \bldg_3|
\;\ge\;
|\bldx\cdot \bldg_1|
\;\ge\;
|\bldx\cdot \bldg_5|
\]
for all \(\bldx\in T\), with equalities only in the boundary cases listed above.
\end{proof}

Since the code is MDS \cite{Roth2026}, we have $h_m(\mathcal{C})=\infty$ for $m=4,5$.
For $m=1,2,3$, the $m$-height optimization reduces to maximizing
\[
f_m(x):=\frac{\bldx\cdot \bldg_0}{d_m(\bldx)},\qquad \bldx\in T,
\]
where the denominator is given by
\[
d_1(\bldx)=\bldx\cdot \bldg_2,\qquad
d_2(\bldx)=\bldx\cdot \bldg_4,\qquad
d_3(\bldx)=-\,\bldx\cdot \bldg_3.
\]

\begin{thm}
\label{thm:icosa_mheight}
Let $\bldx\in T$ and let $f_1,f_2,f_3$ be defined as above. Then the maximizers over $T$ and the corresponding maximal values are:
\begin{enumerate}
  \item $\displaystyle \arg\max_{\bldx\in T} f_1(\bldx)=\bldv_0$, and $f_1(\bldg_0)=\sqrt{5}$.
  \item $\displaystyle \arg\max_{\bldx\in T} f_2(\bldx)=\bldv_0$, and $f_2(\bldg_0)=\sqrt{5}$.
  \item $\displaystyle \arg\max_{\bldx\in T} f_3(\bldx)=\bldv_2$, and $f_3(\bldv_2)=2+\sqrt{5}$.
\end{enumerate}
\end{thm}

\begin{proof}
We prove each claim by analyzing the monotonicity of the corresponding objective on the domain $D$.

\textbf{(1) For $f_1$.}
For $(u,v)\in D$,
\[
f_1(u,v)=\frac{4u+v+\frac{3\sqrt{5}}{2}+\frac{7}{2}}{-2u-2v+\frac{3\sqrt{5}}{2}+\frac{7}{2}}.
\]
A direct differentiation simplifies to
\[
\begin{aligned}
\partial_u f_1(u,v)
&=\frac{-6v+9\sqrt{5}+21}
        {\bigl(-2u-2v+\frac{3\sqrt{5}}{2}+\frac{7}{2}\bigr)^2}
>0,\\
\partial_v f_1(u,v)
&=\frac{6u+\frac{9\sqrt{5}}{2}+\frac{21}{2}}
        {\bigl(-2u-2v+\frac{3\sqrt{5}}{2}+\frac{7}{2}\bigr)^2}
>0.
\end{aligned}
\]
Hence $f_1$ is increasing in both $u$ and $v$ on $D$, so the maximum must lie on the edge $u+v=1$.
Restricting to this edge,
\[
f_1(u,1-u)=\frac{3u+1+\frac{3\sqrt{5}}{2}+\frac{7}{2}}{-2+\frac{3\sqrt{5}}{2}+\frac{7}{2}},
\]
which is strictly increasing in $u$. Therefore the unique maximizer is $u=1$, i.e.,
$(u,v)=(1,0)$, corresponding to $\bldx=\bldv_0=\bldg_0$. Evaluating yields $f_1(\bldg_0)=\sqrt{5}$.

\textbf{(2) For $f_2$.}
For $(u,v)\in D$,
\[
f_2(u,v)=\frac{8u+2v+3\sqrt{5}+7}{-4u+2v+3\sqrt{5}+7}.
\]
Differentiation gives $\partial_u f_2(u,v)>0$ and $\partial_v f_2(u,v)\le 0$ on $D$.
Thus, for any $(u,v)\in D$ we have $f_2(u,0)\ge f_2(u,v)$, and then $f_2(1,0)\ge f_2(u,0)$.
Hence the maximum is attained at $(u,v)=(1,0)$, i.e., $\bldx=\bldv_0=\bldg_0$, and $f_2(\bldg_0)=\sqrt{5}$.

\textbf{(3) For $f_3$.}
For $(u,v)\in D$,
\[
f_3(u,v)=\frac{8u+2v+3\sqrt{5}+7}{(1+\sqrt{5})(2u+2v+1)}.
\]
Since $(1+\sqrt{5})^{-1}>0$, it suffices to maximize
\[
g(u,v)=\frac{8u+2v+3\sqrt{5}+7}{2u+2v+1}.
\]
One checks that $\partial_u g(u,v)<0$ and $\partial_v g(u,v)<0$ on $D$. Hence $g$ (and thus $f_3$)
is strictly decreasing as either $u$ or $v$ increases, so the maximum over $D$ is attained at
$(u,v)=(0,0)$, corresponding to $\bldx=\bldv_2$. Evaluating yields $f_3(\bldv_2)=2+\sqrt{5}$.

This completes the proof.
\end{proof}

\subsection{Dual dodecahedral codes}

We now consider the geometric code derived from the regular dodecahedron. 
The dodecahedron has $20$ vertices, $12$ faces, and $30$ edges. Placing its vertices on the sphere yields $10$ symmetric axes $\{\bldg_0,\ldots,\bldg_{9}\}$ (each axis connects a pair of antipodal vertices), represented by
\[
\resizebox{\columnwidth}{!}{$
G=
\begin{bmatrix}
1 & 1 & 1 & 1 & 0 & 0 & \varphi^{-1} & \varphi^{-1} & \varphi & \varphi \\
1 & 1 & -1 & -1 & \varphi & \varphi & 0 & 0 & \varphi^{-1} & -\varphi^{-1} \\
1 & -1 & 1 & -1 & \varphi^{-1} & -\varphi^{-1} & \varphi & -\varphi & 0 & 0
\end{bmatrix}.
$}
\]

Similar to the icosahedral case, symmetry allows us to restrict attention to a fundamental triangular region on a single dodecahedral face, as illustrated in Fig.~\ref{fig:dode}, with vertices
\[
\bldx_A := \bldg_0,\ \ 
\bldx_B := \frac{\bldg_0+\bldg_4}{2},\ \ 
\bldx_C := \frac{\bldg_0+\bldg_1+\bldg_4+\bldg_5+\bldg_8}{5}.
\]

Let $T'=\mathrm{conv}\{\bldx_A,\bldx_B,\bldx_C\}, \bldx(u,v)=u\,\bldx_A+v\,\bldx_B+(1-u-v)\,\bldx_C,$
where $u,v\ge0$ and $u+v\le1$, and define $\beta_j(x):=|\bldx\cdot \bldg_j|$. Let $\beta_{[j]}(x)$ be the order statistics of $\beta_j(x)$.

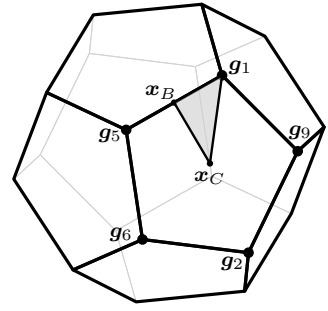
\begin{figure}

\centering

\tdplotsetmaincoords{70}{110}

\begin{tikzpicture}[tdplot_main_coords, line cap=round, line join=round]
  \pgfmathsetmacro{\escala}{1.5}

  \coordinate (1)  at (-\escala*1.37638, \escala*0.,       \escala*0.262866);
  \coordinate (2)  at ( \escala*1.37638, \escala*0.,      -\escala*0.262866);
  \coordinate (3)  at (-\escala*0.425325,-\escala*1.30902, \escala*0.262866);
  \coordinate (4)  at (-\escala*0.425325, \escala*1.30902, \escala*0.262866);
  \coordinate (5)  at ( \escala*1.11352, -\escala*0.809017,\escala*0.262866);
  \coordinate (6)  at ( \escala*1.11352,  \escala*0.809017,\escala*0.262866);
  \coordinate (7)  at (-\escala*0.262866,-\escala*0.809017,\escala*1.11352);
  \coordinate (8)  at (-\escala*0.262866, \escala*0.809017,\escala*1.11352);
  \coordinate (9)  at (-\escala*0.688191,-\escala*0.5,    -\escala*1.11352);
  \coordinate (10) at (-\escala*0.688191, \escala*0.5,    -\escala*1.11352);

  \coordinate (11) at ( \escala*0.688191,-\escala*0.5,     \escala*1.11352);
  \coordinate (12) at ( \escala*0.688191, \escala*0.5,     \escala*1.11352);
  \coordinate (13) at ( \escala*0.850651, \escala*0.,     -\escala*1.11352);
  \coordinate (14) at (-\escala*1.11352, -\escala*0.809017,-\escala*0.262866);
  \coordinate (15) at (-\escala*1.11352,  \escala*0.809017,-\escala*0.262866);
  \coordinate (16) at (-\escala*0.850651, \escala*0.,      \escala*1.11352);
  \coordinate (17) at ( \escala*0.262866,-\escala*0.809017,-\escala*1.11352);
  \coordinate (18) at ( \escala*0.262866, \escala*0.809017,-\escala*1.11352);
  \coordinate (19) at ( \escala*0.425325,-\escala*1.30902, -\escala*0.262866);
  \coordinate (20) at ( \escala*0.425325, \escala*1.30902, -\escala*0.262866);

  \draw[gray!40, ultra thin] (2)  -- (6)  -- (12) -- (11) -- (5)  -- cycle;
  \draw[gray!40, ultra thin] (5)  -- (11) -- (7)  -- (3)  -- (19) -- cycle;
  \draw[gray!40, ultra thin] (11) -- (12) -- (8)  -- (16) -- (7)  -- cycle;
  \draw[gray!40, ultra thin] (12) -- (6)  -- (20) -- (4)  -- (8)  -- cycle;
  \draw[gray!40, ultra thin] (6)  -- (2)  -- (13) -- (18) -- (20) -- cycle;
  \draw[gray!40, ultra thin] (2)  -- (5)  -- (19) -- (17) -- (13) -- cycle;

  \draw[very thick, fill=white, fill opacity=0.15] (4)  -- (20) -- (18) -- (10) -- (15) -- cycle;
  \draw[very thick, fill=white, fill opacity=0.15] (18) -- (13) -- (17) -- (9)  -- (10) -- cycle;
  \draw[very thick, fill=white, fill opacity=0.15] (17) -- (19) -- (3)  -- (14) -- (9)  -- cycle;
  \draw[very thick, fill=white, fill opacity=0.15] (3)  -- (7)  -- (16) -- (1)  -- (14) -- cycle;
  \draw[very thick, fill=white, fill opacity=0.15] (16) -- (8)  -- (4)  -- (15) -- (1)  -- cycle;
  \draw[very thick, fill=white, fill opacity=0.15] (15) -- (10) -- (9)  -- (14) -- (1)  -- cycle;

  \fill (1)  circle[radius=2pt];
  \fill (15) circle[radius=2pt];
  \fill (10) circle[radius=2pt];
  \fill (9)  circle[radius=2pt];
  \fill (14) circle[radius=2pt];

  \node[font=\footnotesize, xshift=7pt,  yshift=3pt] at (1)  {$\bldg_0$};
  \node[font=\footnotesize, xshift=1pt, yshift=7pt] at (15) {$\bldg_8$};
  \node[font=\footnotesize, xshift=-6pt,  yshift=-4pt] at (10) {$\bldg_1$};
  \node[font=\footnotesize, xshift=-8pt, yshift=2pt] at (9)  {$\bldg_5$};
  \node[font=\footnotesize, xshift=-6pt, yshift=-2pt] at (14) {$\bldg_4$};

  \coordinate (xB) at ($(1)!0.5!(14)$);
  \node[font=\footnotesize, xshift=-5pt, yshift=4pt] at (xB) {$\bldx_B$};

  \coordinate (xC) at ($0.2*(1)+0.2*(2)+0.2*(9)+0.2*(5)+0.2*(6)$);
  \node[font=\footnotesize, xshift=24pt, yshift=-1pt] at (xC) {$\bldx_C$};

  \fill (xB)  circle[xshift=0pt, yshift=0pt, radius=1.2pt];
  \fill (xC)  circle[xshift=24pt, yshift=4.5pt, radius=1.2pt];

\coordinate (xBshow) at ($ (xB) + (0pt,0pt) $);
\coordinate (xCshow) at ($ (xC) + (24pt,4.5pt) $);

\fill[gray!35, opacity=0.65] (1) -- (xBshow) -- (xCshow) -- cycle;

\draw[black, line width=0.9pt] (1) -- (xBshow) -- (xCshow) -- cycle;

\end{tikzpicture}
\caption{Dodecahedron with a shaded triangular region indicating the fundamental search space.}
\label{fig:dode}
\end{figure}

\begin{lemma}
\label{lem:dode_ineq}
For any $\bldx \in T'$ we have the following inequalities:
\begin{align}
\beta_0 \ge \beta_4 \ge \beta_8 \ge \max\{\beta_5,\beta_6\}, \label{eq:A1}\\
\beta_5 \ge \beta_1,\qquad \beta_5 \ge \beta_3,\qquad \beta_6 \ge \beta_3, \label{eq:A2}\\
\beta_3 \ge \max\{\beta_2,\beta_7,\beta_9\},\qquad
\beta_1 \ge \max\{\beta_2,\beta_7,\beta_9\}, \label{eq:A3}\\
\beta_2 \le \beta_9. \label{eq:A4}
\end{align}
\end{lemma}

\begin{proof}
We first compute
\begin{align*}
\bldx\cdot \bldg_0
&=\Bigl(2-\frac{2\sqrt5}{5}\Bigr)u
   +\Bigl(\frac12+\frac{\sqrt5}{10}\Bigr)v
   +\Bigl(1+\frac{2\sqrt5}{5}\Bigr),\\
\bldx\cdot \bldg_4
&=\Bigl(-1+\frac{3\sqrt5}{5}\Bigr)u
   +\Bigl(\frac12+\frac{\sqrt5}{10}\Bigr)v
   +\Bigl(1+\frac{2\sqrt5}{5}\Bigr),\\
\bldx\cdot \bldg_8
&=\Bigl(-1+\frac{3\sqrt5}{5}\Bigr)u
   +\Bigl(-\frac12+\frac{\sqrt5}{10}\Bigr)v
   +\Bigl(1+\frac{2\sqrt5}{5}\Bigr),\\
\bldx\cdot \bldg_1
&=\Bigl(-\frac{2\sqrt5}{5}\Bigr)u
   +\Bigl(-\frac{2\sqrt5}{5}\Bigr)v
   +\Bigl(1+\frac{2\sqrt5}{5}\Bigr),\\
\bldx\cdot \bldg_5
&=\Bigl(-\frac{2\sqrt5}{5}\Bigr)u
   +\Bigl(-\frac{2\sqrt5}{5}+\frac{\sqrt5-1}{2}\Bigr)v
   +\Bigl(1+\frac{2\sqrt5}{5}\Bigr),\\
\bldx\cdot \bldg_6
&=\Bigl(\frac{4\sqrt5}{5}\Bigr)u
 +\Bigl(\frac{5+3\sqrt5}{10}\Bigr)v
 +\frac{\sqrt5}{5},\\
\bldx\cdot \bldg_3
&=\Bigl(-1+\frac{\sqrt5}{5}\Bigr)u
  -\Bigl(\frac12+\frac{3\sqrt5}{10}\Bigr)v
  -\frac{\sqrt5}{5},\\
\bldx\cdot \bldg_9
&=\Bigl(1-\frac{\sqrt5}{5}\Bigr)u
  -\frac{\sqrt5}{5}\,v
  +\frac{\sqrt5}{5},\\
\bldx\cdot \bldg_2
&=\Bigl(1+\frac{\sqrt5}{5}\Bigr)u
  +\frac{\sqrt5}{5}\,v
  -\frac{\sqrt5}{5},\\
\bldx\cdot \bldg_7
&=\frac{\sqrt5}{5}
  -\Bigl(1+\frac{\sqrt5}{5}\Bigr)(u+v).
\end{align*}

On $T'$, one checks that $\bldx\cdot \bldg_3 \le 0$ everywhere, while
\[
\bldx\cdot \bldg_j\ge 0\quad \text{for } j\in\{0,1,4,5,6,8,9\}.
\]
The only inner products that may change sign on $T'$ are $\bldx\cdot \bldg_2$ and $\bldx\cdot \bldg_7$.
Therefore, we may drop absolute values for $\bldx\cdot \bldg_j$ with $j\in\{0,1,4,5,6,8,9\}$, keep a minus sign for $|\bldx\cdot \bldg_3|=-(\bldx\cdot \bldg_3)$, and retain absolute values for $|\bldx\cdot \bldg_2|$ and $|\bldx\cdot \bldg_7|$.

\smallskip
\noindent\emph{Proof of \eqref{eq:A1}.}
Compute the differences:
\[
(\bldx\cdot \bldg_0)-(\bldx\cdot \bldg_4)=(3-\sqrt5)\,u\ \ge\ 0,
\ \ \
(\bldx\cdot \bldg_4)-(\bldx\cdot \bldg_8)=v\ \ge\ 0,
\]
\[
\begin{aligned}
&(\bldx\cdot \bldg_8)-(\bldx\cdot \bldg_5)=\Bigl(\frac{5-\sqrt5}{5}\Bigr)u+\frac{\sqrt5-1}{2}\,v \ \ge\ 0,\\
&(\bldx\cdot \bldg_8)-(\bldx\cdot \bldg_6)
=\Bigl(1+\frac{\sqrt5}{5}\Bigr)\bigl(1-u-v\bigr)\ \ge\ 0.
\end{aligned}
\]
which gives $\beta_0\ge\beta_4\ge\beta_8\ge\max\{\beta_5,\beta_6\}$.

\smallskip
\noindent\emph{Proof of \eqref{eq:A2}.}
First,
\[
(\bldx\cdot \bldg_5)-(\bldx\cdot \bldg_1)=\frac{\sqrt5-1}{2}\,v\ \ge\ 0
\quad\Longrightarrow\quad \beta_5\ge \beta_1.
\]
Next,
\[
\beta_5-\beta_3=(\bldx\cdot \bldg_5)+(\bldx\cdot \bldg_3)
=1+\frac{\sqrt5}{5}-\Bigl(1+\frac{\sqrt5}{5}\Bigr)(u+v)\ \ge\ 0,
\]
so $\beta_5\ge\beta_3$. Also,
\[
\beta_6-\beta_3=(\bldx\cdot \bldg_6)+(\bldx\cdot \bldg_3)
=(\sqrt5-1)\,u\ \ge\ 0.
\]
so $\beta_6\ge\beta_3$.

\smallskip
\noindent\emph{Proof of \eqref{eq:A3}.}
We show $|\bldx\cdot \bldg_j|\le \beta_3$ and $|\bldx\cdot \bldg_j|\le \beta_1$ for $j\in\{2,7,9\}$.
For $j=9$ (no absolute needed):
\[
\begin{aligned}
&\beta_3-\beta_{9}=-(\bldx\cdot \bldg_3)-(\bldx\cdot \bldg_{9})
=\frac{1+\sqrt5}{2}\,v\ \ge\ 0,\\
&\beta_1-\beta_{9}
=1+\frac{\sqrt5}{5}-\Bigl(1+\frac{\sqrt5}{5}\Bigr)u-\frac{\sqrt5}{5}\,v\ \ge\ 0.    
\end{aligned}
\]
For $j=2$, it suffices to check $\beta_3\pm(\bldx\cdot \bldg_2)\ge 0$ and $\beta_1\pm(\bldx\cdot \bldg_2)\ge 0$:
\[
\begin{aligned}
&\beta_3-(\bldx\cdot \bldg_2)=\frac{2\sqrt5}{5}(1-u)+\frac{5+\sqrt5}{10}v\ \ge\ 0, \\
&\beta_3+(\bldx\cdot \bldg_2)=2u+\Bigl(\frac12+\frac{\sqrt5}{2}\Bigr)v\ \ge\ 0,\\
&\beta_1-(\bldx\cdot \bldg_2)=1+\frac{3\sqrt5}{5}-\Bigl(1+\frac{3\sqrt5}{5}\Bigr)u-\frac{3\sqrt5}{5}v\ \ge\ 0, \\
&\beta_1+(\bldx\cdot \bldg_2)=1+\frac{\sqrt5}{5}+\Bigl(1-\frac{\sqrt5}{5}\Bigr)u-\frac{\sqrt5}{5}v\ \ge\ 0.
\end{aligned}
\]

For $j=7$, similarly check $\beta_3\pm(\bldx\cdot \bldg_7)\ge 0$ and $\beta_1\pm(\bldx\cdot \bldg_7)\ge 0$:
\[
\begin{aligned}
&\beta_3-(\bldx\cdot \bldg_7)
=2u+\Bigl(\frac12+\frac{3\sqrt5}{10}\Bigr)v+\frac{\sqrt5}{5}\ \ge\ 0,\\
&\beta_3+(\bldx\cdot \bldg_7)
=\frac{2\sqrt5}{5}-\frac{2\sqrt5}{5}u-\frac{5-\sqrt5}{10}\,v \ \ge\ 0,\\
&\beta_1-(\bldx\cdot \bldg_7)
=1+\frac{3\sqrt5}{5}
-\Bigl(1-\frac{\sqrt5}{5}\Bigr)u
-\Bigl(1-\frac{3\sqrt5}{5}\Bigr)v\ \ge\ 0,\\
&\beta_1+(\bldx\cdot \bldg_7)
=1+\frac{\sqrt5}{5}
+\Bigl(1+\frac{\sqrt5}{5}\Bigr)u
+\frac{3\sqrt5}{5}v\ \ge\ 0.
\end{aligned}
\]
Thus $\beta_3\ge \max\{\beta_2,\beta_7,\beta_{9}\}$ and $\beta_1\ge \max\{\beta_2,\beta_7,\beta_{9}\}$.

\smallskip
\noindent\emph{Proof of \eqref{eq:A4}.}
Since $\beta_{9}=\bldx\cdot \bldg_{9}\ge 0$, it suffices to show $\beta_{9}\pm(\bldx\cdot \bldg_2)\ge 0$:
\[
\begin{aligned}
\beta_{9}-(\bldx\cdot \bldg_2)
&=\frac{2\sqrt5}{5}(1-u-v)\ \ge\ 0,\\
\beta_{9}+(\bldx\cdot \bldg_2)
&=2u\ \ge\ 0.
\end{aligned}
\]
hence $|\bldx\cdot \bldg_2|\le \bldx\cdot \bldg_{9}$, i.e., $\beta_2\le \beta_{9}$.
\end{proof}

\begin{lemma}
\label{lem:rank_supports_on_T}
Let $\bldx\in T'$, and let $\beta_{(0)}(\bldx)\ge \cdots \ge \beta_{(9)}(\bldx)$ be the elements of $\{\beta_j(\bldx)\}_{j=0}^{9}$ sorted in non-increasing order.
Under the inequalities \eqref{eq:A1}--\eqref{eq:A4}, the index $j$ satisfying $\beta_j(\bldx)=\beta_{(\ell)}(\bldx)$ can only belong to:
\[
\ell=0:\{0\},\quad
\ell=1:\{4\},\quad
\ell=2:\{8\},\quad
\ell=3:\{5,6\},
\]
\[
\ell=4:\{1,5,6\},\quad
\ell=5:\{1,3,6\},\quad
\ell=6:\{1,3\},
\]
\[
\ell=7:\{7,9\},\quad
\ell=8:\{2,7,9\},\quad
\ell=9:\{2,7\}.
\]
\end{lemma}

\begin{proof}
From \eqref{eq:A1}--\eqref{eq:A4}, we know that ranks
$\ell=0,1,2$ are attained uniquely by $\{0\},\{4\},\{8\}$,
and the remaining ranks are determined by
$\{\beta_1,\beta_2,\beta_3,\beta_5,\beta_6,\beta_7,\beta_{9}\}$.

From \eqref{eq:A3}, we have
\[
\max\{\beta_2,\beta_7,\beta_{9}\}\le \min\{\beta_1,\beta_3\},
\]
so $\{2,7,9\}$ must occupy the bottom three ranks
$\ell=7,8,9$, and $\{1,3,5,6\}$ must occupy ranks
$\ell=3,4,5,6$.

Within $\{1,3,5,6\}$, inequality \eqref{eq:A2} gives
\[
\beta_5\ge \beta_1,\qquad \beta_5\ge \beta_3,\qquad \beta_6\ge \beta_3.
\]
Hence $\beta_1$ cannot be the largest among
$\{\beta_1,\beta_3,\beta_5,\beta_6\}$ since $\beta_5\ge \beta_1$,
and $\beta_3$ cannot be the largest since $\beta_5\ge \beta_3$;
therefore the largest must be attained at index $5$ or $6$, i.e.,
$\ell=3:\{5,6\}$. Moreover, since both $\beta_5$ and $\beta_6$
dominate $\beta_3$, index $3$ cannot be the second-largest, so the
second-largest must lie in $\{1,5,6\}$, i.e., $\ell=4:\{1,5,6\}$.
Next, because $\beta_5\ge \beta_1$ and $\beta_5\ge \beta_3$,
index $5$ cannot be the third- or fourth-largest within this subset;
thus the third-largest must lie in $\{1,3,6\}$, i.e.,
$\ell=5:\{1,3,6\}$. Finally, since $\beta_6\ge \beta_3$, index $6$
cannot be the smallest among the four, so the smallest must lie in
$\{1,3\}$, i.e., $\ell=6:\{1,3\}$.

Within the bottom group $\{2,7,9\}$, inequality \eqref{eq:A4} gives
$\beta_2 \le \beta_{9}$, so index $2$ cannot be the largest among
$\{2,7,9\}$ and index $9$ cannot be the smallest among $\{2,7,9\}$.
Hence
\[
\ell=7:\{7,9\},\qquad
\ell=8:\{2,7,9\},\qquad
\ell=9:\{2,7\}.
\]
\end{proof}

\begin{lemma}
\label{lem:no_stationary_except_3_8}
For each $j\in\{1,3,4,5,6,8,9\}$, define
\[
f_j(u,v):=\frac{\bldx(u,v)\cdot \bldg_0}{|\bldx(u,v)\cdot \bldg_j|},
\qquad \bldx(u,v)\in T'.
\]
Then $f_j$ has no stationary point in $T'$. In fact, on $T'$ we have
\[
\begin{alignedat}{4}
\partial_u f_1&>0,\quad & \partial_v f_1&>0,
&\qquad \partial_u f_4&>0,\quad & \partial_v f_4&\le 0,\\
\partial_u f_5&>0,\quad & \partial_v f_5&>0,
&\qquad \partial_u f_6&<0,\quad & \partial_v f_6&<0,\\
\partial_u f_8&>0,\quad & \partial_v f_8&>0,
&\qquad \partial_u f_9&<0,\quad & \partial_v f_9&>0,
\end{alignedat}
\]
and for $j=3$, $\partial_v f_3<0$.
\end{lemma}
\begin{proof}
On $T'$, we have $\bldx\cdot \bldg_3<0$ and
$\bldx\cdot \bldg_j>0$ for $j\in\{1,4,5,6,8,9\}$.
Hence
\[
|\bldx\cdot \bldg_3|=-(\bldx\cdot \bldg_3),
\qquad
|\bldx\cdot \bldg_j|=\bldx\cdot \bldg_j
\quad j\in\{1,4,5,6,8,9\}.
\]
Write
\[
N:=\bldx\cdot \bldg_0=a_0u+b_0v+c_0,
\qquad
D_j:=|\bldx\cdot \bldg_j|=a_ju+b_jv+c_j.
\]
Then
\[
\partial_u f_j=\frac{a_0D_j-a_jN}{D_j^2},
\qquad
\partial_v f_j=\frac{b_0D_j-b_jN}{D_j^2}.
\]
Substituting the explicit formulas of $\bldx\cdot \bldg_0$ and
$\bldx\cdot \bldg_j$ gives
\begin{align*}
\partial_u f_1
&=\frac{(10+4\sqrt5)+(5-3\sqrt5)v}{5D_1^2}>0,\\
\partial_v f_1
&=\frac{(15+7\sqrt5)+(-10+6\sqrt5)u}{10D_1^2}>0,\\
\partial_u f_4
&=\frac{(5+\sqrt5)+(5-\sqrt5)v}{5D_4^2}>0,\\
\partial_v f_4
&=\frac{(\sqrt5-5)u}{5D_4^2}\le 0,\\
\partial_u f_5
&=\frac{(10+4\sqrt5)+(-5+3\sqrt5)v}{5D_5^2}>0,\\
\partial_v f_5
&=\frac{(5+2\sqrt5)+(5-3\sqrt5)u}{5D_5^2}>0,\\
\partial_u f_6
&=-\frac{2(5+\sqrt5)}{5D_6^2}<0,\\
\partial_v f_6
&=-\frac{5+2\sqrt5}{5D_6^2}<0,\\
\partial_u f_8
&=\frac{(5+\sqrt5)+(-5+\sqrt5)v}{5D_8^2}>0,\\
\partial_v f_8
&=\frac{(5+2\sqrt5)+(5-\sqrt5)u}{5D_8^2}>0,\\
\partial_u f_9
&=-\frac{(5-\sqrt5)+2\sqrt5\,v}{5D_9^2}<0,\\
\partial_v f_9
&=\frac{(5+3\sqrt5)+4\sqrt5\,u}{10D_9^2}>0.
\end{align*}
For $j=3$, we use
\[
\begin{aligned}
D_3
&=|\bldx\cdot \bldg_3|
=-(\bldx\cdot \bldg_3)\\
&=\Bigl(1-\frac{\sqrt5}{5}\Bigr)u
+\Bigl(\frac12+\frac{3\sqrt5}{10}\Bigr)v
+\frac{\sqrt5}{5},
\end{aligned}
\]
and obtain
\[
\partial_v f_3
=-\frac{2\sqrt5\,u+(5+2\sqrt5)}{5D_3^2}<0.
\]
In each case, $\nabla f_j$ cannot be zero on $T'$,
so $f_j$ has no stationary point in $T'$.
\end{proof}

\begin{thm}
\label{thm:dode_m3to7_candidates}
For each $1\le m\le 7$, any maximizer of the corresponding $m$-height over $T'$
must lie in the candidate set
\[
\mathcal{S}
=
\bigl\{(u,v)\in T' :\ 
(1,0),\ (0,0),\ (0,1),
\bigr.
\]
\begin{equation*}
\bigl.
\bigl(0,\tfrac{1+3\sqrt5}{11}\bigr),\
\bigl(\tfrac{\varphi}{3},0\bigr),\ \bigl(0,\,2\sqrt5-4\bigr)
\bigr\}.
\end{equation*}
\end{thm}

\begin{proof}
We treat the cases $m=1$, $m=2$, $m\in\{3,4\}$, $m\in\{5,6\}$, and $m=7$ separately.

\textbf{Case $m=1$.}
From Lemma~\ref{lem:no_stationary_except_3_8}, we have
\[
\partial_u f_4(u,v)
\;>\;0,
\qquad
\partial_v f_4(u,v)
\;<\;0.
\]
Hence $f_4$ is increasing in $u$ and decreasing in $v$ on $T'$,
so the maximizer is attained at $(u,v)=(1,0)$, i.e., $\bldx=\bldx_A=\bldg_0$. Evaluating gives
\[
f_4(\bldg_0)=\frac{\bldg_0\cdot \bldg_0}{\bldg_0\cdot \bldg_4}=\frac{3}{\sqrt5}.
\]

\textbf{Case $m=2$.}
We have
\[
\partial_u f_8(u,v)
\;>\;0,
\qquad
\partial_v f_8(u,v)
\;>\;0.
\]
Thus $f_8$ is increasing in both $u$ and $v$, and by monotonicity the maximizer
lies on the edge $u+v=1$. Along this edge,
\[
f_8(u,1-u)=\frac{(3-\sqrt5)u+(3+\sqrt5)}{(\sqrt5-1)u+(1+\sqrt5)}
\]
is strictly decreasing in $u$, so the unique maximizer is $(u,v)=(0,1)$, i.e.,
$\bldx=\bldx_B=\tfrac{\bldg_0+\bldg_4}{2}$. Evaluating gives
\[
f_8(\bldx_B)=\frac{\bldx_B\cdot \bldg_0}{\bldx_B\cdot \bldg_8}=\varphi.
\]

\medskip
\textbf{Case $m=3,4$ (the denominator is in $\{\bldg_1,\bldg_5,\bldg_6\}$).}
By Lemma~\ref{lem:rank_supports_on_T}, for $m=3,4$, the $(m\!+\!1)$-th largest projection
must be attained at index $1$, $5$, or $6$.

Define
\[
L_{i,j}:=\{\bldx\in T':\ |\bldx\cdot \bldg_i| = |\bldx\cdot \bldg_j|\}.
\]
Since $\bldx\cdot \bldg_1,\ \bldx\cdot \bldg_5,\ \bldx\cdot \bldg_6\ge 0$ on $T'$, we drop the absolute values and write
$L_{i,j}=\{\bldx\in T': \bldx\cdot \bldg_i=\bldx\cdot \bldg_j\}$ for $i,j\in\{1,5,6\}$. The switching boundaries among
these three denominators are given by
\begin{align*}
L_{1,6}:
\bldx\cdot \bldg_1=\bldx\cdot \bldg_6
&\iff 12\sqrt5\,u+(7\sqrt5+5)\,v \\
&=10+2\sqrt5,\\
L_{5,6}:
\bldx\cdot \bldg_5=\bldx\cdot \bldg_6
&\iff 6\sqrt5\,u+(5+\sqrt5)\,v=5+\sqrt5.
\end{align*}
We do not need to consider the switching boundary $L_{1,5}$ since by Lemma~\ref{lem:dode_ineq},
$\bldx\cdot \bldg_5 \ge \bldx\cdot \bldg_1$.

By Lemma~\ref{lem:no_stationary_except_3_8}, none of \(f_1\), \(f_5\), or \(f_6\) has a stationary point in \(T'\).
The switching boundaries partition \(T'\) into three subregions, as illustrated in Fig.~\ref{fig:subregion}.
Within each subregion, the ordering among the denominators is fixed, and the relevant denominator is \(\bldg_j\) for some
\(j\in\{1,5,6\}\).
Consequently, the \(m\)-height optimization reduces to maximizing a single ratio \(f_j(u,v)\) over that subregion.
We thus analyze the three subregions separately, maximizing \(f_1(u,v)\), \(f_5(u,v)\), or \(f_6(u,v)\) on the
corresponding region.

\begin{figure}[htbp]
\centering
\begin{tikzpicture}[scale=5, line cap=round, line join=round]
\usetikzlibrary{calc}

\pgfmathsetmacro{\sqrttwofive}{sqrt(5)}
\pgfmathsetmacro{\uCA}{(5+\sqrttwofive)/(6*\sqrttwofive)} 
\pgfmathsetmacro{\vCB}{(10+2*\sqrttwofive)/(7*\sqrttwofive+5)} 

\coordinate (C) at (0,0);
\coordinate (A) at (1,0);
\coordinate (B) at (0,1);

\coordinate (Q) at (\uCA,0);   
\coordinate (R) at (0,\vCB);   

\fill[blue!18] (C)--(Q)--(R)--cycle;

\fill[green!18] (Q)--(R)--(B)--cycle;

\fill[red!15] (Q)--(A)--(B)--cycle;

\draw[thick] (C)--(A)--(B)--cycle;

\node[below left]  at (C) {$C$};
\node[below right] at (A) {$A$};
\node[above left]  at (B) {$B$};

\draw[thick, dashed]
  (Q)--(B)
  node[pos=0.55, sloped, above] {$L_{5,6}$};

\draw[thick, dashed]
  (Q)--(R)
  node[pos=0.88, sloped, above] {$L_{1,6}$};

\coordinate (M1) at ($(C)!0.5!(Q)$);
\coordinate (T1) at ($(M1)!0.35!(R)$);
\node at (T1) {$T_1$};

\coordinate (M2) at ($(Q)!0.5!(B)$);
\coordinate (T2) at ($(M2)!0.40!(R)$);
\node[xshift=6pt, yshift=-8pt] at (T2) {$T_2$};

\coordinate (M3) at ($(A)!0.5!(B)$);
\coordinate (T3) at ($(M3)!0.33!(Q)$);
\node at (T3) {$T_3$};

\end{tikzpicture}
\caption{Triangle $T'$ in $(u,v)$-coordinates and switching lines $L_{1,6}$ and $L_{5,6}$, partitioning $T'$ into three subregions.}
\label{fig:subregion}
\end{figure}

\smallskip
\noindent\textbf{(A) Maximizers in $T_1$.}
From Lemma~\ref{lem:no_stationary_except_3_8}, both $\partial_u f_1$ and $\partial_v f_1$
are strictly positive. Hence $f_1$ is strictly increasing in both $u$ and $v$, and its
maximum over $T_1$ must be attained on the edge $L_{1,6}$. Evaluating $f_1$ along this edge shows that
\[
f_1\Bigl(u,\frac{10+2\sqrt5-12\sqrt5\,u}{7\sqrt5+5}\Bigr)
=\frac{(13-5\sqrt5)u+(13+6\sqrt5)}{(4\sqrt5-6)u+(5+4\sqrt5)},
\]
with derivative
\[
\frac{d}{du}(\cdot)=
\frac{11(\sqrt5-7)}{\bigl((4\sqrt5-6)u+(5+4\sqrt5)\bigr)^2}<0.
\]
Hence the maximum is attained at $u=0$, i.e.,
\[
(u,v)=\Bigl(0,\ \frac{1+3\sqrt5}{11}\Bigr).
\]

For $f_5$, both $\partial_u f_5$ and $\partial_v f_5$ are strictly positive. Evaluating $f_5$ along $L_{1,6}$ shows that
\[
f_5\bigl(u,v(u)\bigr)
=\frac{(6\sqrt5-10)u+(25+11\sqrt5)}{(4\sqrt5-20)u+(15+9\sqrt5)}.
\]
Differentiating yields
\[
\frac{d}{du}f_5\bigl(u,v(u)\bigr)
=\frac{40(10+3\sqrt5)}{\bigl((4\sqrt5-20)u+(15+9\sqrt5)\bigr)^2}>0.
\]
Hence $f_5$ is strictly increasing along $L_{1,6}$ inside $T'$, so its maximum on this segment
is attained at the endpoint with largest $u$, namely where $v=0$:
\[
(u,v)=\Bigl(\frac{10+2\sqrt5}{12\sqrt5},\,0\Bigr)
=\Bigl(\frac{1+\sqrt5}{6},\,0\Bigr).
\]

For $f_6$, both partial derivatives satisfy $\partial_u f_6<0$ and $\partial_v f_6<0$ on $T'$, so
$f_6$ is strictly decreasing in both variables. Hence its maximum over $T_1$ is attained at the
vertex $(u,v)=(0,0)$.

\bigskip
\noindent\textbf{(B) Maximizers in $T_2$.}
For $f_1$, both $\partial_u f_1$ and $\partial_v f_1$ are strictly positive. Evaluating $f_1$ along $L_{5,6}$ shows that
\[
f_1\Bigl(u,\ 1+\tfrac32(1-\sqrt5)u\Bigr)
=\frac{(4-2\sqrt5)u+(3+\sqrt5)}{2\bigl((3-\sqrt5)u+1\bigr)}.
\]
Differentiating gives
\[
\frac{d}{du}(\cdot)=
\frac{-4\sqrt5}{\bigl(2((3-\sqrt5)u+1)\bigr)^2}<0,
\]
hence the maximum is attained at $u=0$, i.e., at $(u,v)=(0,1)$.

For $f_5$, both $\partial_u f_5$ and $\partial_v f_5$ are strictly positive. Evaluating $f_5$ along $L_{5,6}$ shows that
\[
f_5\Bigl(u,\ 1+\tfrac32(1-\sqrt5)u\Bigr)
=\frac{(4-2\sqrt5)u+(3+\sqrt5)}{(\sqrt5-3)u+(1+\sqrt5)}.
\]
with derivative
\[
\frac{d}{du}(\cdot)=
\frac{10(1+\sqrt5)}{\bigl((\sqrt5-3)u+(1+\sqrt5)\bigr)^2}>0.
\]
Thus the maximum on $L_{5,6}$ is attained at the largest feasible $u$, which occurs at $v=0$, i.e.,
\[
(u,v)=\Bigl(\frac{1+\sqrt5}{6},\ 0\Bigr).
\]

For $f_6$, both $\partial_u f_6$ and $\partial_v f_6$ are strictly negative. Evaluating $f_6$ along $L_{1,6}$ shows that
\[
f_6\Bigl(u,\frac{10+2\sqrt5-12\sqrt5\,u}{7\sqrt5+5}\Bigr)
=\frac{(6\sqrt5-10)u+(25+11\sqrt5)}{(10-2\sqrt5)u+(15+5\sqrt5)}.
\]
Its derivative is
\[
\frac{d}{du}(\cdot)
=-\frac{20(7+\sqrt5)}{\bigl((10-2\sqrt5)u+(15+5\sqrt5)\bigr)^2}<0.
\]
Hence the maximum on this segment is attained at the endpoint with smallest $u$, i.e.,
\[
(u,v)=\Bigl(0,\ \frac{1+3\sqrt5}{11}\Bigr).
\]

\bigskip
\noindent\textbf{(C) Maximizers in $T_3$.}
For $f_1$, both $\partial_u f_1$ and $\partial_v f_1$ are strictly positive. Evaluating $f_1$ along edge $AB$ shows that
\[
f_1(u,1-u)
=\bigl(\frac32-\frac{\sqrt5}{2}\bigr)u+\bigl(\frac32+\frac{\sqrt5}{2}\bigr).
\]
Differentiating yields
\[
\frac{d}{du}(\cdot)=\frac32-\frac{\sqrt5}{2}>0,
\]
so $f_1$ is strictly increasing along $AB$.
Hence the maximum on $AB$ is attained at $u=1$, i.e., at the vertex $A=(1,0)$.

For $f_5$, both $\partial_u f_5$ and $\partial_v f_5$ are strictly positive. Evaluating $f_5$ along edge $AB$ shows that
\[
f_5(u,1-u)
=\frac{(3-\sqrt5)u+(3+\sqrt5)}{(1-\sqrt5)u+(1+\sqrt5)}.
\]
Its derivative is
\[
\frac{d}{du}f_5(u,1-u)
=\frac{4\sqrt5}{\bigl((1-\sqrt5)u+(1+\sqrt5)\bigr)^2}>0,
\]
so $f_5(u,1-u)$ is strictly increasing on $u$. Therefore the maximum on edge $AB$
is attained at $u=1$.

For $f_6$, both $\partial_u f_6$ and $\partial_v f_6$ are strictly negative. Evaluating $f_6$ along edge $L_{5,6}$ shows that
\[
f_6\Bigl(u,\,1+\tfrac{3}{2}(1-\sqrt5)u\Bigr)
=\frac{(2-\sqrt5)u+\tfrac{3+\sqrt5}{2}}
{\tfrac{\sqrt5-3}{2}\,u+\tfrac{1+\sqrt5}{2}}.
\]
Its derivative is
\[
\frac{d}{du}f_6\Bigl(u,\,1+\tfrac{3}{2}(1-\sqrt5)u\Bigr)
=\frac{\tfrac{\sqrt5-1}{2}}
{\bigl(\tfrac{\sqrt5-3}{2}\,u+\tfrac{1+\sqrt5}{2}\bigr)^2}>0.
\]
so $f_6$ is strictly increasing along $L_{5,6}$.
Therefore, the maximum on this segment is attained at the endpoint with largest $u$,
namely
\[
(u,v)=\Bigl(\frac{1+\sqrt5}{6},\,0\Bigr).
\]

\medskip
\textbf{Case $m=5,6$ (the denominator is in $\{\bldg_1,\bldg_3,\bldg_6\}$).}
For $m=5$ or $m=6$, the $(m\!+\!1)$-th largest magnitude must be attained at index $1$, $3$, or $6$. On $T'$,
$\bldx\cdot \bldg_1\ge0$ and $\bldx\cdot \bldg_6\ge0$, while $\bldx\cdot \bldg_3<0$, so
$|\bldx\cdot \bldg_1|=\bldx\cdot \bldg_1$, $|\bldx\cdot \bldg_6|=\bldx\cdot \bldg_6$, and
$|\bldx\cdot \bldg_3|=-(\bldx\cdot \bldg_3)$.
The switching interfaces are
\begin{align*}
L_{1,3}:\ &2(5+\sqrt5)\,u+(5+7\sqrt5)\,v=10+2\sqrt5,\\
L_{1,6}:\ &12\sqrt5\,u+(7\sqrt5+5)\,v=10+2\sqrt5.
\end{align*}
Similar to the previous case, these lines partition $T'$ into three subregions.

\begin{figure}[htbp]
\centering
\begin{tikzpicture}[scale=5, line cap=round, line join=round]
\usetikzlibrary{calc}

\pgfmathsetmacro{\sqrttwofive}{sqrt(5)}
\pgfmathsetmacro{\uQ}{(10+2*\sqrttwofive)/(12*\sqrttwofive)} 
\pgfmathsetmacro{\vR}{(10+2*\sqrttwofive)/(7*\sqrttwofive+5)}  

\coordinate (C) at (0,0);
\coordinate (A) at (1,0);
\coordinate (B) at (0,1);

\coordinate (Q) at (\uQ,0);   
\coordinate (R) at (0,\vR);   

\fill[blue!18] (C)--(Q)--(R)--cycle;

\fill[green!18] (A)--(Q)--(R)--cycle;

\fill[red!15] (A)--(R)--(B)--cycle;

\draw[thick] (C)--(A)--(B)--cycle;

\node[below left]  at (C) {$C$};
\node[below right] at (A) {$A$};
\node[above left]  at (B) {$B$};

\draw[thick, dashed]
  (A)--(R)
  node[pos=0.55, sloped, above] {$L_{1,3}$};

\draw[thick, dashed]
  (Q)--(R)
  node[pos=0.5, sloped, above] {$L_{1,6}$};

\coordinate (M1) at ($(C)!0.5!(Q)$);
\coordinate (T1) at ($(M1)!0.38!(R)$);
\node at (T1) {$T_1$};

\coordinate (M2) at ($(A)!0.5!(Q)$);
\coordinate (T2) at ($(M2)!0.50!(R)$);
\node[xshift=20pt, yshift=-20pt] at (T2) {$T_2$};

\coordinate (M3) at ($(A)!0.5!(B)$);
\coordinate (T3) at ($(M3)!0.28!(R)$);
\node at (T3) {$T_3$};

\end{tikzpicture}
\caption{Triangle $T'$ in $(u,v)$-coordinates and switching lines $L_{1,3}$ and $L_{1,6}$, partitioning $T'$ into three subregions.}
\label{fig:subregion_case2}
\end{figure}

\smallskip
\noindent\textbf{(A) Maximizers in $T_1$.}
Maximizing $f_2$ and $f_7$ in $T_1$ is already done in case I, thus we skip them.

For $f_4$, From Lemma~\ref{lem:no_stationary_except_3_8}, $\partial_v f_4 < 0$. Thus the maximizer is on $v=0$, which is edge $CA$
\[
f_4(u,0)
=\frac{(10-2\sqrt5)u+(5+2\sqrt5)}{(5-\sqrt5)u+\sqrt5}.
\]
Its derivative is
\[
\frac{d}{du}f_4(u,0)
=\frac{5(\sqrt5-5)}{\bigl((5-\sqrt5)u+\sqrt5\bigr)^2}<0.
\]
Therefore $f_4(u,0)$ is strictly decreasing on $u\in[0,1]$, and the maximum on the edge $v=0$
is attained at $u=0$.

\smallskip
\noindent\textbf{(A) Maximizers in $T_1$.}
Maximizing $f_1$ and $f_6$ in $T_1$ was already done in the case $m=3,4$, so we skip them.

For $f_3$, from Lemma~\ref{lem:no_stationary_except_3_8}, $\partial_v f_3 < 0$. Thus the maximizer is on $v=0$, which is edge $CA$.
Along this edge,
\[
f_3(u,0)
=\frac{(10-2\sqrt5)u+(5+2\sqrt5)}{(5-\sqrt5)u+\sqrt5}.
\]
Its derivative is
\[
\frac{d}{du}f_3(u,0)
=\frac{5(\sqrt5-5)}{\bigl((5-\sqrt5)u+\sqrt5\bigr)^2}<0.
\]
Therefore $f_3(u,0)$ is strictly decreasing on $u\in[0,1]$, and the maximum on the edge $v=0$
is attained at $u=0$.

\smallskip
\noindent\textbf{(B) Maximizers in $T_2$.}
For $f_1$, both $\partial_u f_1$ and $\partial_v f_1$ are strictly positive. Evaluating $f_1$ along edge $L_{1,3}$ shows that
\[
\begin{aligned}
&f_1\Bigl(u,\frac{10+2\sqrt5-2(5+\sqrt5)\,u}{5+7\sqrt5}\Bigr) \\
&\qquad =
\frac{(-10+10\sqrt5)\,u+(25+11\sqrt5)}
     {(-10+2\sqrt5)\,u+(15+5\sqrt5)}.
\end{aligned}
\]
Its derivative is
\[
\frac{d}{du}\,f_1(u,v(u))
=\frac{240+160\sqrt5}{\bigl((-10+2\sqrt5)\,u+(15+5\sqrt5)\bigr)^2}>0.
\]
Therefore, the maximum is attained at the endpoint with the largest $u$, i.e., at the vertex $A$.

For \(f_3\), we have \(\partial_v f_3<0\). Hence any maximizer must lie on \(CA\cap T_2\) or on \(L_{1,6}\).
From the previous calculation, along \(CA\) we also have \(\partial_u f_3<0\). Therefore, if the maximizer lies on \(CA\cap T_2\),
it must occur at the endpoint where \(CA\cap T_2\) meets \(L_{1,6}\). Consequently, it suffices to restrict attention to \(L_{1,6}\).
Evaluating $f_3$ along edge $L_{1,6}$ shows that
\[
f_3\Bigl(u,\frac{10+2\sqrt5-12\sqrt5\,u}{7\sqrt5+5}\Bigr)
=\frac{(13-5\sqrt5)u+(13+6\sqrt5)}{(5-7\sqrt5)u+(5+4\sqrt5)}.
\]
Its derivative is
\[
\frac{d}{du}\,f_3
=\frac{110+88\sqrt5}{\bigl((5-7\sqrt5)u+(5+4\sqrt5)\bigr)^2}>0.
\]
Therefore the maximum of $f_3$
over $L_{1,6}$ is attained at the endpoint where $v=0$, i.e.,
\[
(u,v)=\Bigl(\frac{1+\sqrt5}{6},\,0\Bigr).
\]

For $f_6$, both $\partial_u f_6$ and $\partial_v f_6$ are strictly negative. Thus the maximum is on edge $L_{1,6}$,
which has already been calculated in the case $m=3,4$.

\smallskip
\noindent\textbf{(C) Maximizers in $T_3$.}
The maximizer for $f_1$ in $T_3$ is on edge $AB$, and maximizing $f_1$ along $AB$ was already done in the case $m=3,4$.

For $f_3$, since $\partial_v f_3 < 0$, the maximizer is on edge $L_{1,3}$. Since on that edge, $f_1 = f_3$,
it is equivalent to evaluating $f_1$ along $L_{1,3}$, which was already done above.

For $f_6$, since $\partial_v f_6 < 0$, the maximizer is on edge $L_{1,3}$.
\[
\begin{aligned}
&f_6\Bigl(u,\frac{10+2\sqrt5-2(5+\sqrt5)\,u}{5+7\sqrt5}\Bigr) \\
&\qquad =
\frac{(10\sqrt5-10)\,u+(25+11\sqrt5)}
     {20u+(15+5\sqrt5)}.
\end{aligned}
\]
Differentiating gives
\[
\frac{d}{du}\,f_6
=-\frac{40(10+3\sqrt5)}{\bigl(20u+(15+5\sqrt5)\bigr)^2}<0.
\]
Therefore, the maximum of $f_6$ over
$L_{1,3}$ is attained at $u=0$, i.e.,
\[
(u,v)=\Bigl(0,\ \frac{1+3\sqrt5}{11}\Bigr).
\]

\medskip
\textbf{Case $m=7$ (the denominator is in $\{\bldg_7,\bldg_{9}\}$).}

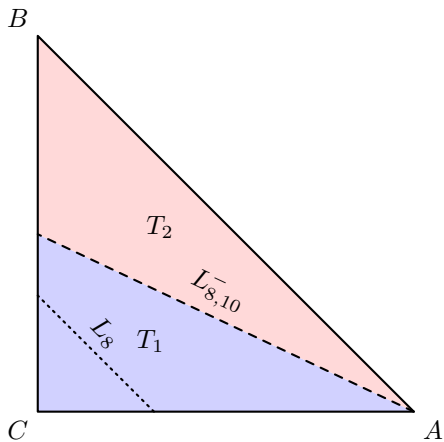
\begin{figure}[htbp]
\centering
\begin{tikzpicture}[scale=5, line cap=round, line join=round]
\usetikzlibrary{calc}

\pgfmathsetmacro{\sqrttwofive}{sqrt(5)}
\pgfmathsetmacro{\k}{2*\sqrttwofive-4}           
\pgfmathsetmacro{\cL}{(\sqrttwofive-1)/4}        

\coordinate (C) at (0,0);
\coordinate (A) at (1,0);
\coordinate (B) at (0,1);

\coordinate (R) at (0,\k); 

\coordinate (P) at (\cL,0); 
\coordinate (S) at (0,\cL); 

\fill[blue!18] (C)--(A)--(R)--cycle;

\fill[red!15] (A)--(R)--(B)--cycle;

\draw[thick] (C)--(A)--(B)--cycle;

\node[below left]  at (C) {$C$};
\node[below right] at (A) {$A$};
\node[above left]  at (B) {$B$};

\draw[thick, dashed]
  (A)--(R)
  node[pos=0.55, sloped, above] {$L_{7,9}^{-}$};

\draw[thick, dotted]
  (P)--(S)
  node[pos=0.55, sloped, above] {$L_{7}$};

\coordinate (M1) at ($(C)!0.5!(A)$);
\coordinate (T1) at ($(M1)!0.40!(R)$);
\node at (T1) {$T_1$};

\coordinate (M2) at ($(A)!0.5!(B)$);
\coordinate (T2) at ($(M2)!0.35!(R)$);
\node at (T2) {$T_2$};

\end{tikzpicture}
\caption{Triangle $T'$ in $(u,v)$-coordinates, the switching line $L_{7,9}^{-}$, and the zero line $L_7:\ x\cdot g_7=0$.}
\label{fig:subregion_case_L810m_with_L8}
\end{figure}

The switching condition $|\bldx\cdot \bldg_7|=|\bldx\cdot \bldg_{9}|$ splits into
\[
L_{7,9}^{+}: \ \bldx\cdot \bldg_7=\bldx\cdot \bldg_{9},
\qquad
L_{7,9}^{-}: \ \bldx\cdot \bldg_7=-\bldx\cdot \bldg_{9},
\]
which reduce to
\[
L_{7,9}^{+}:\ v=-2u,
\qquad
L_{7,9}^{-}:\ v=(2\sqrt5-4)(1-u).
\]
The switching line $L_{7,9}^{+}:\ v=-2u$ intersects $T'$ only at the vertex $C=(0,0)$.
Hence, it suffices to consider the two subregions partitioned by
\[
L_{7,9}^{-}:\ v=(2\sqrt5-4)(1-u).
\]
We further denote by $L_7$ the zero line of $\bldx\cdot \bldg_7$, i.e., $\bldx\cdot \bldg_7=0$.

In subregion $T_1$, the ordering satisfies
\[
|\bldx\cdot \bldg_{9}|\ge |\bldx\cdot \bldg_7|.
\]
Moreover, since $\bldx\cdot \bldg_{9}\ge 0$ throughout $T'$, the $m$-height function
for $m=7$ reduces to
\[
f_{9}(u,v)=\frac{\bldx\cdot \bldg_0}{\bldx\cdot \bldg_{9}}.
\]
By Lemma~\ref{lem:no_stationary_except_3_8}, we have
\[
\partial_u f_{9}<0,
\qquad
\partial_v f_{9}>0
\quad \text{on } T'.
\]
Therefore, the maximum of $f_{9}$ over $T_1$ is attained at the vertex
with minimal $u$ and maximal $v$, namely at the intersection of $L_{7,9}^{-}$ and
the edge $BC$, which is
\[
(u,v)=(0,\ 2\sqrt5-4).
\]

In subregion $T_2$, the ordering is reversed and the $m$-height function is
\[
f_{7}(u,v)=\frac{\bldx\cdot \bldg_0}{-\,\bldx\cdot \bldg_7}.
\]
Then direct differentiation yields
\[
\begin{aligned}
\partial_u f_7
&=\frac{\Bigl(1-\frac{\sqrt5}{5}\Bigr)v-(1+\sqrt5)}{D(u,v)^2}<0,\\
\partial_v f_7
&=\frac{\Bigl(\frac{\sqrt5}{5}-1\Bigr)u-\Bigl(\frac32+\frac{7\sqrt5}{10}\Bigr)}{D(u,v)^2}<0.
\end{aligned}
\]
Hence $f_7$ is strictly decreasing in both variables, and its maximum over $T_2$
must be attained on the boundary $L_{7,9}^{-}$.

Evaluating $f_7$ along $L_{7,9}^{-}:\ v=(2\sqrt5-4)(1-u)$ gives
\[
f_7\Bigl(u,\ (2\sqrt5-4)(1-u)\Bigr)
=\frac{(3-\sqrt5)u+\sqrt5}{(3-\sqrt5)u+(\sqrt5-2)}.
\]
Differentiating,
\[
\begin{aligned}
\frac{d}{du}\,f_7\Bigl(u,\ (2\sqrt5-4)(1-u)\Bigr)
&=
\frac{2(\sqrt5-3)}
     {\bigl((3-\sqrt5)u+(\sqrt5-2)\bigr)^2} \\
&<0.
\end{aligned}
\]
so this restriction is strictly decreasing in $u$. Consequently, the maximum of $f_7$ on
$T_2$ is attained at the endpoint $u=0$, i.e.,
\[
(u,v)=(0,\ 2\sqrt5-4).
\]
This completes the proof.
\end{proof}

The following table summarizes the $m$-height profiles of the dual polyhedral codes.

\begin{table}[htbp]
\centering
\caption{$m$-height profiles of the dual polyhedral codes.}
\label{tab:mheight_dual_polyhedral}
\begin{tabular}{c c c}
\hline
Code & $m$ & $m$-height\\
\hline
\multirow{3}{*}{Dual icosahedral}
& $1$ & $\sqrt{5}$\\
& $2$ & $\sqrt{5}$\\
& $3$ & $2+\sqrt{5}$\\
\hline
\multirow{7}{*}{Dual dodecahedral}
& $1$ & $\dfrac{3}{\sqrt{5}}$\\
& $2$ & $\varphi$\\
& $3$ & $4-\sqrt{5}$\\
& $4$ & $3$\\
& $5$ & $2+\sqrt{5}$\\
& $6$ & $2+\sqrt{5}$\\
& $7$ & $5+2\sqrt{5}$\\
\hline
\end{tabular}
\end{table}

\end{document}